%% file: main.tex
\documentclass[sigconf, nonacm, pdfa]{acmart}

\newcommand\vldbdoi{10.14778/3773749.3773751}

\newcommand\vldbavailabilityurl{https://github.com/JdHondt/MS-Index}
\newcommand\vldbpagestyle{empty} 

\usepackage{ifthen}
\usepackage{amsmath}
\usepackage{amsfonts}
\usepackage{algorithmic}
\usepackage{graphicx}
\usepackage{textcomp}
\usepackage{xcolor}
\usepackage{xspace}
\usepackage{verbatim}
\usepackage{soul}
\usepackage{url}
\usepackage{hyperref}
\usepackage{subcaption}
\usepackage{caption}
\usepackage{makecell}
\usepackage{float}
\usepackage{longtable}
\usepackage{bm}
\usepackage{pdflscape}
\usepackage[]{mdframed}
\usepackage{multirow}
\usepackage{balance}
\usepackage{tikz}
\usepackage{enumitem}
\usepackage{titlesec}
\usetikzlibrary{shapes,math}
\usepackage{pgfplots} 
\usepackage{pifont}
\usepackage{booktabs, makecell}
\usepackage[tworuled, linesnumbered, noend]{algorithm2e}
\usepackage{pgfplots}
\pgfplotsset{compat=1.16}
\usepgfplotslibrary{groupplots}
\usetikzlibrary{fit,backgrounds}
\hypersetup{colorlinks=true,urlcolor=blue,citecolor=red}
\SetKw{Continue}{continue}
\SetKw{Break}{break}

\def\mode{paper}

\theoremstyle{definition}

\newcommand\rev[2]{\protect \textcolor{red}{[#1] #2}}
\newcommand\todo[1]{\protect \textcolor{green}{[TODO] #1}}
\newcommand\cut[1]{}
\newcommand\newcut[1]{\st{#1}}
\newcommand\ody[1]{\protect \textcolor{red}{[ODY] #1}}
\newcommand\jdh[1]{\protect \textcolor{blue}{[ODY] #1}}
\newcommand\tp[1]{\protect \textcolor{orange}{[TP] #1}}
\renewcommand\vec[1]{\bm{#1}}
\newcommand\avec[1]{\tilde{\bm{#1}}} 
\newcommand\knn{$k$-NN\xspace}
\newcommand\q{\vec{Q}}
\newcommand\qlen{|\q|}
\newcommand\name{MS-Index\xspace}

\newcommand\dstree{DSTree*\xspace}
\newcommand\stindex{ST-Index*\xspace}
\newcommand\kvmatch{KV-Match*\xspace}

\renewcommand{\abstractname}{\textsc{ABSTRACT}}
\renewcommand{\refname}{\textsc{REFERENCES}}

\ifthenelse{\equal{\mode}{paper}}{
    \renewcommand\rev[2]{\protect #2}
    \renewcommand\todo[1]{}
    \renewcommand\cut[1]{}
    \renewcommand\newcut[1]{}
    \renewcommand\ody[1]{}
    \renewcommand\jdh[1]{}
    \renewcommand\tp[1]{}
}{}

\newcommand{\STAB}[1]{\begin{tabular}{@{}c@{}}#1\end{tabular}}
\renewcommand{\rotcell}[1]{\STAB{\rotatebox[origin=c]{90}{#1}}}

\begin{document}
\title{\name: Fast Top-k Subsequence Search for Multivariate Time Series under Euclidean Distance}

\twocolumn

\input{sections/0.authors.tex}
\input{sections/0.abstract.tex}

\maketitle

\pagestyle{\vldbpagestyle}


\setcounter{page}{1}

\input{sections/1.introduction}
\input{sections/2_.preliminaries}
\input{sections/2.related_work.tex}
\input{sections/3.methodology}
\input{sections/4.baselines}
\input{sections/5.evaluation}
\input{sections/6.conclusions}
\input{sections/n.acknowledgements.tex}

\balance
\bibliographystyle{ACM-Reference-Format}
\bibliography{references}

\end{document}

%% file: sections/0.authors.tex
\author{Jens d'Hondt}
\email{j.e.d.hondt@tue.nl}
\affiliation{%
\institution{Eindhoven University of Technology}
\city{Eindhoven}
\country{the Netherlands}
}

\author{Teun Kortekaas}
\email{teun@jelter.net}
\affiliation{%
\institution{Eindhoven University of Technology}
\city{Eindhoven}
\country{the Netherlands}
}

\author{Odysseas Papapetrou}
\email{o.papapetrou@tue.nl}
\affiliation{%
\institution{Eindhoven University of Technology}
\city{Eindhoven}
\country{the Netherlands}
}

\author{Themis Palpanas}
\email{themis@mi.parisdescartes.fr}
\affiliation{%
\institution{Universit\'e Paris Cit\'e \& IUF}
\city{Paris}
\country{France}
}

%% file: sections/0.abstract.tex
\begin{abstract}
Modern applications frequently collect and analyze temporal data in the form of multivariate time series (MTS) -- time series that contain multiple channels. 
A common task in this context is subsequence search, which involves identifying all MTS that contain subsequences highly similar to a query time series.
In practical scenarios, not all channels of an MTS are relevant to every query. 
For instance, airplane sensors may gather data on a plethora of components and subsystems, but only a few of these are relevant to a specific query, such as identifying the cause of a malfunctioning landing gear, or a specific flight maneuver.
Consequently, the relevant query channels are often specified at query time.
In this work, we introduce the \emph{Multivariate Subsequence Index} (\name), a novel algorithm for nearest neighbor MTS subsequence search under Euclidean distance that supports ad-hoc selection of query channels. 
The algorithm is \rev{R1D1}{\textit{exact}} and demonstrates query performance that scales sublinearly to the number of query channels. 
We examine the properties of \name with a thorough experimental evaluation over 34 datasets, and show that it outperforms the state-of-the-art one to two orders of magnitude for both raw and normalized subsequences.
\end{abstract}

%% file: sections/1.introduction.tex
\input{figures/airplane.tex}

\section{INTRODUCTION}\label{sec:introduction}
Time series are ubiquitous in diverse domains, such as astrophysics, seismology, meteorology, health care, finance, video and audio recordings~\cite{bach-andersen17,huijse14,kashino99,knieling17,paraskevopoulos13}.
Due to advances in sensor technology, the amount of time series being collected has increased dramatically over the last years, particularly in the form of multivariate time series (MTS)~\cite{roddick01}. Each MTS is a collection of time series, often sourced from different sensors that measure different aspects of the same phenomenon or object.
These different time series are referred to as the \textit{channels} of the MTS.
Examples of MTS include health monitoring data of patients in a hospital (e.g., heart rate, blood pressure, and body temperature), climate sensor arrays (e.g., an array measuring the temperature, humidity, and air pressure, at a certain location)~\cite{liess2017}, or motion capture data (e.g., the position and acceleration of different body parts).

A key operation on time series is similarity search, which involves finding the most similar time series to a query time series (also known as its \emph{Nearest Neighbors})~\cite{hydra_exact,hydra_approx}.
Similarity search can be performed as a standalone task~\cite{isax2,dallachiesa14,faloutsos94,keogh04,kondylakis18,linardi18_ulisse,rafiei98,ucr_suite,shieh08}, but is also frequently used as a subroutine in tasks such as outlier detection~\cite{breunig00,dallachiesa14,tuli22}, classification~\cite{bagnall18,hills14,ruiz21,shifaz23,yang04,yekta14}, and clustering~\cite{bagnall04,bonifati23,ding15,kalpakis01,salarpour18}.

To get the most out of the ever-growing datasets, modern-day similarity search algorithms should support queries with high flexibility.
For example, when analyzing patient data, a doctor might need to find similar occurrences of a short-term pattern in a patient's vital signs, to understand the root cause of their condition.
Such a query would require the algorithm to support (a) MTS (multiple sensors), (b) comparison of subsequences rather than  whole time series, and (c) deciding the query channels at query time (only the sensors relating to the patient's relevant vital signs).
Another example requiring such queries is the analysis of sensor data from airplanes~\cite{bhaduri10}, where the user is analyzing the failed landing of an airplane due to a malfunctioning landing gear, and wants to find similar occurrences in historical data to understand the root cause.
In this case, the user might select the period of time of the failed landing (i.e., the left time series in Figure~\ref{fig:airplane}) as well as the relevant channels (highlighted red in the figure) to find similar occurrences in historical data.

While much work has been done to address the challenges that come with similarity search on large datasets of univariate time series (UTS), the work on MTS search is still rudimentary.
Current solutions only support searching for whole time series (whole-matching) instead of subsequences (subsequence search), and only on a fixed, pre-determined set of channels.
In this work, we show that the performance of whole-matching state-of-the-art solutions suffers when these are extended for subsequences.
A further extension of these solutions to handle multivariate data only exacerbates the problem, due to the additional challenges that come with the number of channels. 
In summary, existing algorithms fall short in at least one of the following ways: (a) they natively support only UTS, and their extension to MTS is non-trivial, (b) they are focused on whole matching, and their performance becomes unacceptable when inserting subsequences, or, (c) they rely on severely restrictive assumptions of the user query, such as user-defined thresholds on each channel.

In this work, we propose \emph{\name}, a novel algorithm for k-nearest neighbor (\knn) subsequence search on MTS under Euclidean distance.
\name allows for selection of the query channels at query time, and works for both normalized and un-normalized subsequences.\footnote{Supporting normalized subsequences is more challenging than normalized time series, as it cannot be done through preprocessing the data. Furthermore, while most applications require normalized subsequences for shape matching, some applications require querying for raw subsequences to preserve scale differences~\cite{palpanas_itisa,palpanas_dagstuhl}.}
\name supports \rev{R1D5}{\emph{fixed-length queries}}, i.e., it requires the length of the query to be known at index construction. 
The algorithm is \rev{R1D1}{\emph{exact}}: it always returns the correct result.

\name differs from existing methods in three ways. 
First, it is designed specifically for MTS, leveraging the additional pruning potential that comes with multiple channels. 
Second, it combines index-powered pruning of the search space with efficient distance computation on the remaining candidates through the convolution theorem (MASS)~\cite{conv_theorem,mueen20}.
This is done with the help of an R-tree that indexes Discrete Fourier Transform (DFT) approximated subsequences \emph{on all channels}, and a two-pass search algorithm that prunes candidates based on the lower-bound distance to the query.
\rev{R1D3}{
    To the best of our knowledge, this is the first time an index has been combined with MASS for subsequence search; previous works have only used each of these techniques in isolation~\cite{faloutsos94,mueen20}, arguing for either an index-based or a sequential-scan approach.
}
\rev{R1D3}{
    Third, it utilizes two novel methods to tighten the lower bound distance to the query, further improving the pruning potential of the search.
    These methods are \emph{generic}, meaning that they can be applied to any search solution that uses R-trees and/or DFT approximations.
}
Our key contributions are as follows:
\begin{itemize}[leftmargin=0.4cm,topsep=3pt]
    \item We introduce \name, a novel algorithm for \knn subsequence search on MTS under Euclidean distance (Section~\ref{sec:methodology}).
    \item \rev{R1D3}{We propose a general set of optimizations to improve any search solution using R-trees and/or DFT approximations (Section~\ref{sec:optimizations}). These are shown to improve the performance of \name by a factor of 4, when combined.}
    \item We provide a general approach to extending current solutions for UTS to the multivariate case, which provide baselines for our evaluation (Section~\ref{sec:baseline}).
    \item We conduct a thorough evaluation of \name across 34 datasets, comparing it to a range of baselines, and showing that \name outperforms the state-of-the-art by two orders of magnitude for both raw and normalized subsequences (Section~\ref{sec:evaluation}).
\end{itemize}



%% file: figures/airplane.tex
\begin{figure}[t]
    \centering
    \begin{tikzpicture}
        
       \begin{groupplot}[
          group style={
             group name=plot,
             group size=2 by 4,
             xlabels at=edge bottom,
             xticklabels at=edge bottom,
             vertical sep=.3cm,
             horizontal sep=1cm
          },
          width=.5\columnwidth,
          height=70pt,
          grid=both,
       ]
       \nextgroupplot[
            width=2.7cm,
             align=center,
             title={Query $\vec{Q}$},
             ylabel={\textcolor{red}{Altitude}},
             xmin=0, xmax=500,
             ymin=0, ymax=16000,
             axis background/.style={fill=gray!20} 
       ]
           \addplot[color=blue, mark=*, smooth] coordinates {
             (0, 10000)
             (250, 5000)
             (500, 14000)
           };
           \coordinate (c11topleft) at (axis cs:0,\pgfkeysvalueof{/pgfplots/ymax});
           \coordinate (c11topright) at (axis cs:500,\pgfkeysvalueof{/pgfplots/ymax});
           \coordinate (c11botleft) at (axis cs:0,\pgfkeysvalueof{/pgfplots/ymin});
           \coordinate (c11botright) at (axis cs:500,\pgfkeysvalueof{/pgfplots/ymin});
      \nextgroupplot[
            width=\columnwidth-2.7cm,
            align=center,
            title={1NN: Plane 287 - 23-04-2001},
            xmin=9500, xmax=11250,
            ymin=0, ymax=30000,
            axis background/.style={fill=gray!20} 
      ]
         \addplot[color=blue, mark=*, smooth] coordinates {
            (9000, 28500)
            (9250, 26500)
            (9500, 22500)
            (9750, 18500)
            (10000, 14500)
            (10250, 10500)
            (10500, 5500)
            (10750, 13500)
            (11000, 18700)
            (11250, 21200)
         };
         \coordinate (c21topleft) at (axis cs:10250,\pgfkeysvalueof{/pgfplots/ymax});
          \coordinate (c21topright) at (axis cs:10750,\pgfkeysvalueof{/pgfplots/ymax});
          \coordinate (c21botleft) at (axis cs:10250,\pgfkeysvalueof{/pgfplots/ymin});
          \coordinate (c21botright) at (axis cs:10750,\pgfkeysvalueof{/pgfplots/ymin});
          
       \nextgroupplot[
             width=2.7cm,
             ylabel={Latitude},
             xmin=0, xmax=500,
             ymin=37.7, ymax=37.85,
       ]
          \addplot[color=blue, mark=*, smooth] coordinates {
            (0, 37.76)
            (250,37.79)
            (500,37.79)
          };
       \nextgroupplot[
         width=\columnwidth-2.7cm,
        xmin=9500, xmax=11250,
         ymin=52.360, ymax=52.375,
         ytick={52.360, 52.370},
         yticklabels={52.36, 52.37},
       ]
          \addplot[color=blue, mark=*, smooth] coordinates {
            (9000, 52.365)
            (9250, 52.368)
            (9500, 52.367)
            (9750, 52.368)
            (10000, 52.37)
            (10250, 52.369)
            (10500, 52.367)
            (10750, 52.365)
            (11000, 52.365)
            (11250, 52.365)
          };
 
      \nextgroupplot[
          width=2.7cm,
          align=center,
          ylabel={\textcolor{red}{Landing}\\\textcolor{red}{Gear}},
          xlabel={Time},
          xmin=0, xmax=500,
          xtick={0, 250, 500}, 
          xticklabels={0, 250, 500},
          ymin=0, ymax=1,
          ytick={0, 1},
          yticklabels={0, 1},
          axis background/.style={fill=gray!20} 
        ]
          \addplot[ycomb,color=blue, mark=*] coordinates {
            (0, 0)
            (125, 1)
            (250, 0)
            (375, 1)
            (500, 0)
          };
         \coordinate (c12topleft) at (axis cs:0,\pgfkeysvalueof{/pgfplots/ymax});
          \coordinate (c12topright) at (axis cs:500,\pgfkeysvalueof{/pgfplots/ymax});
          \coordinate (c12botleft) at (axis cs:0,\pgfkeysvalueof{/pgfplots/ymin});
          \coordinate (c12botright) at (axis cs:500,\pgfkeysvalueof{/pgfplots/ymin});

       \nextgroupplot[
         width=\columnwidth-2.7cm,
            xmin=9500, xmax=11250,
            xlabel={Time},
            xtick={10000, 10500, 11000, 11500},
            xticklabels={10000, 10500, 11000, 11500},
            ymin=0, ymax=1,
            ytick={0, 1},
            yticklabels={0, 1},
            axis background/.style={fill=gray!20} 
       ]
       \addplot[ycomb,color=blue, mark=*] coordinates {
          (9000, 0)
          (9125, 0)
          (9250, 0)
          (9375, 0)
          (9500, 0)
          (9625, 0)
          (9750, 0)
          (9875, 0)
          (10000, 0)
          (10125, 0)
          (10250, 0)
          (10375, 1)
          (10500, 0)
          (10625, 1)
          (10750, 0)
          (10875, 0)
          (11000, 0)
          (11125, 0)
          (11250, 0)
       };
        \coordinate (c22topleft) at (axis cs:10250,\pgfkeysvalueof{/pgfplots/ymax});
        \coordinate (c22topright) at (axis cs:10750,\pgfkeysvalueof{/pgfplots/ymax});
        \coordinate (c22botleft) at (axis cs:10250,\pgfkeysvalueof{/pgfplots/ymin});
        \coordinate (c22botright) at (axis cs:10750,\pgfkeysvalueof{/pgfplots/ymin});
       \end{groupplot}      
      
      \fill[red!70!white, draw=black, fill opacity=0.3] (c11topleft) rectangle (c11topright) rectangle (c11botleft) rectangle (c11botright);
      \fill[red!70!white, draw=black, fill opacity=0.3] (c12topleft) rectangle (c12topright) rectangle (c12botleft) rectangle (c12botright);
      
      \fill[red!70!white, draw=black, fill opacity=0.3] (c21topleft) rectangle (c21topright) rectangle (c21botleft) rectangle (c21botright);
      \fill[red!70!white, draw=black, fill opacity=0.3] (c22topleft) rectangle (c22topright) rectangle (c22botleft) rectangle (c22botright);


      \draw[dashed, line width=1.5pt] (c21topleft) -- (c22botleft);
      \draw[dashed, line width=1.5pt] (c21topright) -- (c22botright);
     \end{tikzpicture}
    \caption{Example query and 1NN for MTS of synthetic airplane data. Altitude and landing gear are the query channels. The highlighted boxes (red) are the considered subsequences.}
    \label{fig:airplane}
 \end{figure}
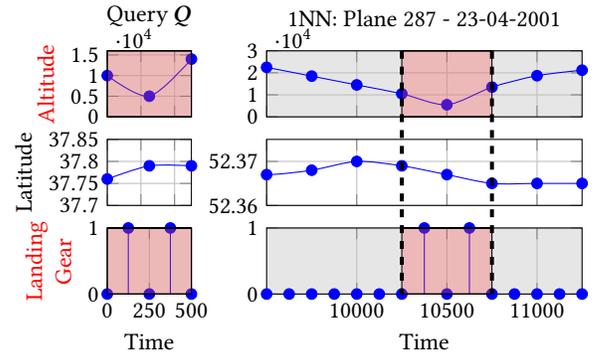

%% file: sections/2_.preliminaries.tex
\section{PRELIMINARIES}\label{sec:preliminaries}
\subsection{Definitions and Notation}\label{sec:definitions}
\paragraph{Time series.}
A UTS $\vec{t}$ of length $m$ is denoted as $\vec{t} = [\vec{t}_1, \ldots, \vec{t}_m]$, where $\vec{t}_{i}$ is the data point at time $i$. 
An MTS $\vec{T}$ with $c$ channels is denoted as a matrix $\vec{T} = [\vec{t}_1, \ldots, \vec{t}_c]^T$, where $\vec{t}_i$ is the UTS of channel $i$. 
A subsequence of length $s$ starting at timepoint $i$ is denoted as $\vec{T}_i^s = \vec{T}_{*,i:i+s-1}$.
The number of available channels in the dataset is denoted as $c$, and the set of indices to the query channels is denoted as $c_{\q}$, with $c_{\q}\subseteq \{1,\dots,c\}$.


\paragraph{Similarity Search Queries.}
There are two forms of similarity search queries; a k-Nearest-Neighbor (\knn) query and an r-range query (also called threshold query).
A \knn query is defined as finding the $k$ time series in a dataset $\mathcal{D}$ with the smallest distances to a query time series $\q$, under a given distance measure $d$~\cite{hydra_exact}.
In the case of MTS, all time series in $\mathcal{D}$ contain the same $c$ channels.
However, the query time series $\q$ may contain only a subset of these channels $c_{\q}$.
Conversely, an r-range query is defined as finding all time series in $\mathcal{D}$ with a distance to a query time series $\q$ smaller than a given threshold $r$~\cite{hydra_exact}. 
These query types can be further categorized into whole-matching and subsequence search queries. 
In \emph{whole-matching}, we consider the distance between an entire query time series and an entire candidate series~\cite{hydra_exact}. 
All the series involved in the search need to be of the same length.
In \emph{subsequence search}, we consider the distance between an entire query series and all subsequences of a candidate series with the same length as the query~\cite{hydra_exact}.
In this case, the candidate series from which the subsequences are extracted need not have the same length.
\rev{R1D5}{There exist two variants of subsequence search; (a) \emph{fixed-length} search, where the query length $\qlen$ is predetermined and fixed across queries, and (b) \emph{variable-length} search, where the query length is not fixed and can vary between queries~\cite{hydra_exact}.
In this work, we focus on the problem of fixed-length subsequence search, which is a common assumption in the literature~\cite{faloutsos94,dstree,isax2,vlachos03,bhaduri10}.
Variable-length search is left for future work.
}

\paragraph{Distance.}
In line with previous studies on MTS similarity search~\cite{bagnall18,vlachos03,bhaduri10}, we use Euclidean distance (ED) to measure the distance between time series, which is defined over MTS as: 
\begin{equation}\label{eq:euclidean}
    d(\vec{X},\vec{Y}) = ||\vec{X}-\vec{Y}||_2 = \sqrt{\sum_{i \in c_{\vec{X}\vec{Y}}} \sum_{j=1}^{m} (\vec{X}_{i,j} - \vec{Y}_{i,j})^2}
\end{equation}
with $c_{\vec{X}\vec{Y}} = c_{\vec{X}} \cap c_{\vec{Y}}$ the common channels of $\vec{X}$ and $\vec{Y}$, and $m = \min(m_{\vec{X}},m_{\vec{Y}})$ the length of the shortest time series.
\rev{R1D6}{ED was chosen due to its simplicity, efficiency, well-understood properties for univariate data, its natural extension to the multivariate case, and its popularity in the literature~\cite{bagnall18,vlachos03,bhaduri10,sigmod25}.
In fact, while temporal alignment through more complex measures like Dynamic Time Warping (DTW) and Shape-Based Distance (SBD) has proven effective for whole-matching in both univariate and multivariate contexts~\cite{dhondt24,salarpour18,shifaz23,paparrizos15,paparrizos23,sigmod25}, the marginal gain of using these measures over ED for subsequence search is negligible~\cite{keogh_dtwmyths05,keogh_dtwmyths06,shieh08,sigmod25}.
This is because the consideration of all subsequences in time series is effectively a form of temporal alignment, comparing the query with candidate time series under different shifts, naturally correcting for temporal distortions.}
Table~\ref{tab:nomenclature} summarizes the notations used throughout the paper.

\begin{table}[t]
    \centering
    \small
    \caption{Nomenclature}
    \label{tab:nomenclature}
    \begin{tabular}{ll}
        \hline
        $\vec{t}$ & Univariate time series \\
        $\vec{T}$ & Multivariate time series \\
        $|\vec{t}|$ & Length of time series $\vec{T}$ \\
        $\vec{T}_i$ & Channel $i$ of $\vec{T}$  \\
        $\vec{T}_{i,j}$ & Data point $j$ of channel $i$ of $\vec{T}$ \\
        $\vec{T}_i^s$ & Subsequence of $\vec{T}$ of length $s$ starting at index $i$ \\
        $\avec{T}$ & DFT approximation of $\vec{T}$ \\
        $\avec{T}'$ & Flattened DFT approximation of $\vec{T}$ (i.e. feature vector) \\
        $\q$ & Query time series \\
        $c$ & Total number of channels in the dataset \\
        $c_Q$ & The set of channel ids of $\q$ \\
        $n$ & Number of time series in the dataset $\mathcal{D}$ \\
        $m$ & Length of the time series in $\mathcal{D}$ \\
        \hline
    \end{tabular}
\end{table}


\begin{figure}[t]
    \centering
    \includegraphics[width=\columnwidth]{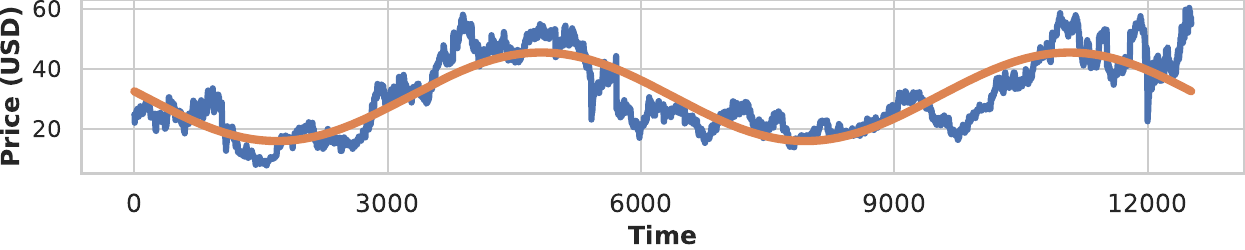}
    \caption{The price of a stock over time (blue), and its reconstruction through its first three DFT coefficients (orange).}
    \label{fig:fourier_example}
\end{figure}

\subsection{DFT Approximation}\label{sec:dft}
A key challenge in similarity search is the high dimensionality of the data, sourcing from the length of the time series.
A common approach to address this challenge is the use of dimensionality reduction, or ``summarization'', techniques that transform the data into a lower-dimensional space where distances can be approximated at a low cost. 
A popular summarization technique for Euclidean distance ($L_2$) is the DFT approximation, which involves performing a DFT on the time series and keeping only a small fraction of the resulting vector~\cite{faloutsos94,mueen10,rafiei98,strang94}.
The DFT decomposes a time series into a sum of sinusoids of different frequencies (also called the \emph{coefficients}), where the amplitudes of the sinusoids represent the importance of the corresponding frequency in the time series (also called the \emph{energy} of the coefficient).
\cut{Particularly, the DFT transforms a time series from the time-domain to the frequency-domain, resulting in a complex vector of the same length, with each value corresponding to a sinusoid of given frequency, in ascending order.}
For most real-world time series, the amplitudes of the high-frequency sinusoids are small, which means that these sinusoids can be discarded without losing much information.
This means that we can accurately approximate a time series using only the first $f$ values of the DFT, resulting in a vector of size $f$ instead of the original size $m$~\cite{faloutsos94}.
The accuracy of such an approximation is demonstrated with an example from the \emph{Stocks} dataset (cf. Section~\ref{sec:evaluation}) in Figure~\ref{fig:fourier_example}.
DFT approximation is popular for estimating the Euclidean distance, as the distance between two DFT-approximated time series is a lower bound on their true Euclidean distance~\cite{mueen10}.
Namely, the Euclidean distance between two time series $\vec{t},\q \in \mathbb{R}^m$ is bounded by their DFT approximations $\avec{t},\avec{Q} \in \mathbb{C}^f$ as~\cite{mueen10}:
\begin{equation}\label{eq:dft_bound_uts}
d(\vec{t},\q) \geq \sqrt{\frac{\sum_{i=1}^f \vert\avec{t}_i - \avec{Q}_i\vert^2}{m}} = \frac{d(\avec{t},\avec{Q})}{\sqrt{m}}
\end{equation}
As the value of $f$ increases (i.e., we account for more coefficients), the DFT bound converges to the exact distance~\cite{mueen10}.
This way, DFT approximation allows for a trade-off between the accuracy of the distance estimation and the dimensionality of the data, which is particularly useful for high-dimensional data such as time series.
In the case of \rev{R2D3}{MTS}, the same concept can be used by summarizing each individual channel, independently.

\subsection{R-tree Construction Techniques}\label{sec:rtree}
An R-tree is a tree-based index for spatial data that groups nearby objects into Minimum Bounding Rectangles (MBRs), which are recursively grouped into larger MBRs~\cite{guttman84,rstar}. 
R-trees can be constructed top-down by splitting a single MBR containing all data points into smaller MBRs using strategies like Quadratic Split~\cite{guttman84}, Linear Split~\cite{guttman84}, or R*-tree topological split~\cite{rstar}. 
Alternatively, bottom-up construction (or \emph{bulk loading}) starts with individual data points as MBRs and merges them into larger MBRs based on a \emph{leaf size} $L$ and partitioning strategy, generally resulting in less overlap but requiring all data to be known in advance.
A popular bulk loading algorithm is the \emph{Sort-Tile-Recursive} (STR) algorithm~\cite{leutenegger97}.
STR first sorts the children based on their coordinates in each dimension (using the middle in case the children are rectangles instead of points).
Then, it recursively partitions the entries into groups of size $\lceil\frac{N}{L}\rceil^{1/d}$, where $N$ is the number of entries to index, $L$ is the desired leaf size, and $d$ is the dimensionality of the tree.
For example, given a 1D space with entries $[1,2,\dots,10]$ and $L=2$, the algorithm partitions the entries into $[[1,2],[3,4],[5,6],[7,8],[9,10]]$, to subsequently create the nodes $[[1,\dots,4], [5,\dots 8], [9,10]]$, etc.

%% file: sections/2.related_work.tex
\subsection{Related Work}\label{sec:related} 
We first discuss related work on UTS, followed by recent work on MTS. 
An overview of all related work is given in Table~\ref{tab:related}.

\begin{table}[t]
    \centering
    \footnotesize
    \caption{Overview of related work}
    \label{tab:related}
    \begin{tabular}{ll|llll}
        \cline{3-6} & & \textbf{Name} & \textbf{Ref} & \textbf{Type} & \textbf{Notes} \\ \hline
        \multicolumn{1}{l|}{\multirow{6}{*}{\rotcell{UTS}}} & \multirow{3}{*}{\rotcell{Whole}}       & VA+ file      & \cite{vafile}       & Index         &                                          \\
        \cline{3-6}
        \multicolumn{1}{l|}{}                     &                              & ISAX2+        & \cite{isax2}        & Index         &                                          \\
        \cline{3-6}
        \multicolumn{1}{l|}{}                     &                              & DSTree        & \cite{dstree}       & Index         &                                          \\\cline{2-6}
        \multicolumn{1}{l|}{}                     & \multirow{4}{*}{\rotcell{Sub.}} & ST-index      & \cite{faloutsos94}  & Index         &                                          \\ \cline{3-6}
        \multicolumn{1}{l|}{}                     &                              & KV-match      & \cite{kvmatch}      & Sequential         & \thead[l]{Custom definition of \\normalized subsequences} \\
        \cline{3-6}
        \multicolumn{1}{l|}{}                     &                              & MASS          & \cite{mueen20}      & Sequential    &                                          \\
        \hline
        \multicolumn{1}{l|}{\multirow{5}{*}{\rotcell{MTS}}} & \rotcell{   Whole  }                        & Vlachos       & \cite{vlachos03}    & Index         & \thead[l]{Does not support\\ normalized subsequences} \\ \cline{2-2}
        \cline{2-6}
        \multicolumn{1}{l|}{}                     & \multirow{2}{*}{\rotcell{Sub.}}                  & Bhaduri       & \cite{bhaduri10}    & Index         & \thead[l]{Only supports\\ r-range queries}            \\ \cline{3-6}
        \multicolumn{1}{l|}{}                     &                   & \name (ours)       & & Index         &      \\ \hline
        \end{tabular}
\end{table}

\paragraph{Univariate Time Series (UTS) whole-matching.}
In a large-scale comparison of algorithms for UTS whole-matching by Echihabi, et al.~\cite{hydra_exact}, index-based algorithms generally showed to outperform their sequential scan counterparts.
Furthermore, for small in-memory datasets, VA+ file~\cite{vafile} and iSAX2+~\cite{isax2} excel, while DSTree~\cite{dstree} dominates for larger datasets. 
DSTree uses Extended Adaptive Piecewise Constant Approximation (EAPCA) to estimate distances between query and groups of time series, differentiating them by mean and variance at increasing resolution. 
Whole-matching indices could be extended to support subsequence search by indexing each subsequence individually.
However, this approach would overwhelm the indices when $\qlen \ll |\vec{T}|$, degrading performance as verified in Section~\ref{sec:evaluation}.

\paragraph{UTS subsequence search indices.}
Multiple indices have been proposed for UTS subsequence search.
The ST-index~\cite{faloutsos94} for r-range queries first extracts a DFT approximation of each subsequence of length $\qlen$ in a time series.
These $f$-dimensional vectors then form a trail in the $f$-dimensional space, which is segmented into sub-trails using MBRs.
The MBRs are then indexed in an R*-tree~\cite{rstar}, to enable efficient search.  
More recents improvements through Dual Match and General Match~\cite{gmatch,dmatch} used different window types to reduce index size. However, these extensions inherently restrict the indices to raw subsequences, as supporting normalized subsequences would require a complete redesign.
Our work adopts the idea of indexing DFT approximations in an R-tree but (a) focuses on MTS, (b) removes explicit time series segmentation, (c) leverages convolution theorem for faster distance computation, and (d) employs a different search algorithm to support \knn queries with \emph{ad-hoc} selection of query channels.
Another notable index is KV-match~\cite{kvmatch}, which indexes subsequences within a single time series using the means and variances of disjoint windows over the time series.
The index supports both normalized and raw subsequences, though for normalized subsequences it restricts the search space to subsequences with similar means and variances to the query \emph{before normalization}.
KV-Match can be adapted to our problem setting by (a) building one KV-match per time series and iteratively querying each, and (b) dropping the filter on means and variances before querying normalized subsequences.

Other UTS subsequence indices include L-match~\cite{feng20} and TS-index~\cite{chatzi23}, which we do not consider further because: L-match improves KV-match's indexing time but has higher query time; and TS-index uses Chebyshev distance rather than our Euclidean distance.
As we show in Section~\ref{sec:evaluation}, our work outperforms KV-match, which (by transivity) also suggests how our work is expected to relate to L-match.

\paragraph{UTS subsequence search sequential scans.} 
Mueen's Algorithm for Similarity Search (MASS)~\cite{mueen20} is an exact subsequence search algorithm that computes the distance between a query $\q$ and all subsequences of a time series $\vec{t}$ in time $O(|\vec{t}| \log |\vec{t}|)$ rather than the exhaustive $O(|\vec{t}|\qlen)$ time. 
It does so through the convolution theorem, which states that the cross-correlation (i.e., sliding dot product) of two time series is equivalent to the point-wise multiplication of their Fourier transforms~\cite{conv_theorem}.
Namely, the convolution theorem states that the dot-products $\langle \cdot, \cdot \rangle$ between $\q$ and all subsequences of $\vec{t}$ of length $\qlen$ can be computed through:
\begin{equation}
    [\langle\q,\vec{t}_1^{\qlen}\rangle, \ldots, \langle\q,\vec{t}_{|\vec{t}|-\qlen+1}^{\qlen}\rangle] = \mathcal{F}^{-1}(\mathcal{F}(\q) \otimes \mathcal{F}(\vec{t}))
\label{eq:conv_theorem}
\end{equation}
where $\mathcal{F}$, $\mathcal{F}^{-1}$, and $\otimes$ are the Fourier transform, inverse Fourier transform, and point-wise multiplication, respectively. 
Then, as the Euclidean distance between two vectors $\vec{x}$ and $\vec{y}$ is defined as $d(\vec{x},\vec{y})=\sqrt{||\vec{x}||^2 + ||\vec{y}||^2 - 2\langle\vec{x},\vec{y}\rangle}$, MASS computes the distance between $\q$ and all subsequences of $\vec{t}$ in time $O(|\vec{t}| \log |\vec{t}|)$ using Equation~\ref{eq:conv_theorem} and the sliding squared sums of $\q$ and $\vec{t}$.
MASS can be used as a subsequence search algorithm for MTS, by repeating the distance computation for each channel separately and summing the results to obtain the multivariate Euclidean distance.
We will be using MASS in this work, both as a baseline and as a component of our method.

\paragraph{MTS whole-matching.}
The first index for MTS whole-matching was proposed by Vlachos et al.~\cite{vlachos03}, focusing on r-range queries under DTW and Longest Common Subsequence (LCSS) distance, besides Euclidean distance.
It works by splitting the MTS into MBRs that span across the time, channel, and value axes, and storing those in an R*-tree~\cite{rstar}.
Then, at query time, a \emph{Minimum Bounding Envelope} (MBE) is constructed for the query, which covers all the possible matching areas of the query under warping conditions.
This MBE is further decomposed into MBRs and probed in the R*-tree to efficiently find the candidates.
However, using the method for subsequence search under Euclidean distance essentially reduces to a simplified variant of Dual Match~\cite{dmatch}, which -- as discussed earlier -- prevents the index from supporting normalized subsequences, making it unsuitable for our problem setting.

\paragraph{MTS subsequence search.}
\rev{R1D7}{
    Recently, a novel algorithm for \textit{variable-length} MTS subsequence search was proposed named MULISSE~\cite{mulisse}, as a multivariate extension of the state-of-the-art UTS subsequence search algorithm ULISSE~\cite{linardi18_ulisse}.
    The method extends iSAX2+~\cite{isax2} representations to multiple channels and introduces channel-aware node splitting for improved pruning.
    To accommodate for the vast search space of variable-length subsequence search, the ULISSE index (and thus MULISSE) uses a lightweight design that sacrifices pruning power for a reduced index size.
    MULISSE can be applied to fixed-length search by restricting the query length range to a single value, as we show in Section~\ref{sec:evaluation}.
}

To the best of our knowledge, the only solution for \textit{fixed-length} MTS subsequence search is by Bhaduri, et al~\cite{bhaduri10}.
The method answers \emph{r-range} queries under Euclidean distance by indexing subsequences per channel based on distances to univariate reference points.
However, its reliance on channel-level thresholds prevents support for \knn queries, which require simultaneous querying of all channels.
Adapting the method to our problem setting would require non-trivial modifications likely to degrade its performance, so we do not consider it further in our work.

\paragraph{Summary.}
Concluding, literature includes a multitude of similarity search solutions for UTS and MTS data under different query types and solution designs.
However, only few of the existing solutions can be directly applied/extended to our problem setting.
Namely, all MTS solutions either (a) do not support \knn queries~\cite{bhaduri10} or querying on normalized subsequences~\cite{vlachos03}, \rev{R1D7}{or (b) optimize for index size rather than query performance (MULISSE~\cite{mulisse}).}
There are UTS solutions that could be extended to MTS, but those extensions are not expected to scale well.
This is either because (a) they are not build for subsequence search (\emph{DSTree}), (b) they are sequential scan algorithms (\emph{MASS}), or (c) they are simply designed for UTS (\emph{ST-index}).

%% file: sections/3.methodology.tex
\section{MS-INDEX}\label{sec:methodology}
MTS subsequence search involves two key challenges. 
First, the number of subsequences in the dataset grows rapidly, even for fairly small datasets. 
For a dataset with $n$ time series, each of length $m$, an exhaustive search would require computing the Euclidean distance between all $n*(m-\qlen+1)$ subsequences and the query $\vec{Q}$, with each computation costing $|c_{\vec{Q}}|*\qlen$ time.
Second, the requirement to support ad-hoc selection of query channels precludes the use of summarizations that involve merging channels, such as PCA~\cite{yang04}.

Our solution, named \name, addresses these challenges by combining the pruning potential of an index with the efficiency of the MASS algorithm for exact distance computation (cf. Section~\ref{sec:related}). 
The solution involves three key steps (cf. Figure~\ref{fig:summarization}): (a) summarizing the MTS subsequences through DFT approximations, (b) building an R-tree index on these approximations,
\footnote{\rev{R1D1}{We would like to stress that, despite the use of DFTs, which is an approximation technique, \name remains an \textit{exact} algorithm. 
Approximation only influences the bounding efficiency and pruning power (i.e., the speed) of \name, and not the accuracy or completeness of its results.
Specifically, after applying dimensionality reduction, the distances between objects in the R-tree are lower-bounds on the actual distances. This potentially leads to false positives, but never false negatives as we query the R-tree with a threshold that is an upper bound on the \emph{actual} \knn distance (i.e., without  dimensionality reduction). We formally prove this in Lemma~\ref{lem:correctness}.}} and (c) at query time, using the index to heavily prune the candidate subsequences, followed by an exact distance computation on the remaining candidates using MASS.
Through the index, \name is able to prune over 99\% of the candidate subsequences for a variety of datasets, at a very low cost, thereby achieving a speedup up to 100x compared to the state-of-the-art. 
We also propose a number of optimizations to further improve the efficiency of the search process (Section~\ref{sec:optimizations}).
These novel optimizations are interesting in their own right, as they can be applied to any search algorithm that uses DFT approximations and/or spatial indices, such as~\cite{ciaccia97,mueen10,faloutsos94,rafiei98,vlachos03,vafile}.

We now present the key components of \name, starting with time series summarization (Section~\ref{sec:summarization}) and indexing (Section~\ref{sec:indexing}), followed by the query execution algorithm (Section~\ref{sec:query}), and the optimizations (Section~\ref{sec:optimizations}).

\begin{figure}[t]
    \centering
    \includegraphics[width=1.05\linewidth]{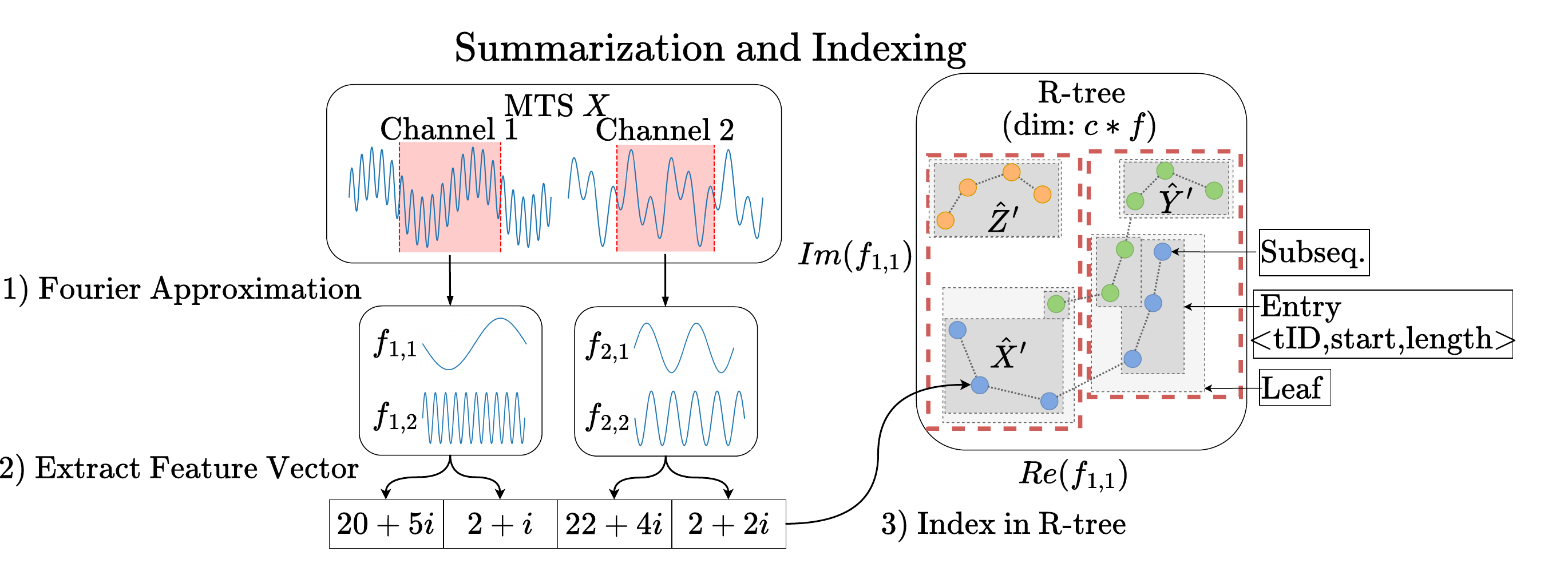}
    \caption{\rev{R1D4}{Summarization and indexing of MTS subsequences. We represent the indexed subsequences of three MTS in the feature space with blue, green, and orange nodes. Time-neighbouring subsequences are connected with lines.}}
    \label{fig:summarization}
\end{figure}

\subsection{Summarizing All Subsequences with DFTs}
\label{sec:summarization}
As discussed in Section~\ref{sec:dft}, a big challenge in working with time series is their length.
Using a subset of DFT coefficients to approximate time series has been extensively used in the literature~\cite{faloutsos94,mueen10,vafile}. 
This approach is effective because it provides a compact representation that enables efficient lower-bounding of distances between time series.
However, our analysis of real-world time series revealed two important observations that can further improve the accuracy of such approximations;

\noindent\textbf{Observation 1:} The \textit{first}-$f$ coefficients are not always the \textit{top}-$f$ coefficients in terms of their contribution to the total energy of the time series; there exist domains where certain high-frequency contributions are also substantial.
This is illustrated in the two plots in Figure~\ref{fig:cumdist_and_query}a, which show the absolute and  cumulative contribution to the total energy and distance across DFT coefficients, extracted from 100 time series with real-world temperature readings at different locations (see Section~\ref{sec:evaluation} for details on the dataset).
Both lines show a 90\% confidence interval over a wide range of measurements (i.e., all 100 time series to form the line for energy, and all 100*99/2 pairwise distances to form the line for total distance). 
The figure shows a sudden jump in the energy contribution at the 62nd coefficient, which indicates that this coefficient has a substantial contribution to the shape of the time series and should thus be included in the approximation.
The top-$f$ most significant coefficients can be derived similarly to the configuration process in other adaptive summarization techniques~\cite{grosse12,paparrizos19grail,spartan}; by computing statistics on a sample of the dataset pre-index construction. 

\begin{figure*}[tbh]
    \begin{subfigure}[t]{0.35\textwidth}
        \centering
        \includegraphics[width=\linewidth]{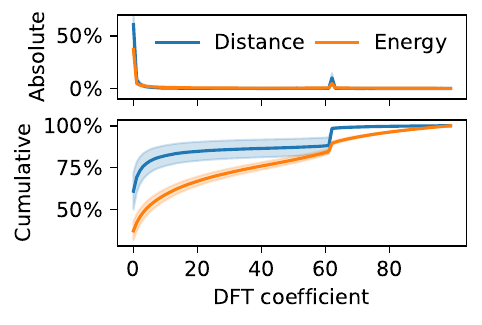}
    \end{subfigure}
    \hfill
    \begin{subfigure}[t]{0.62\textwidth}
        \centering
        \includegraphics[width=\linewidth]{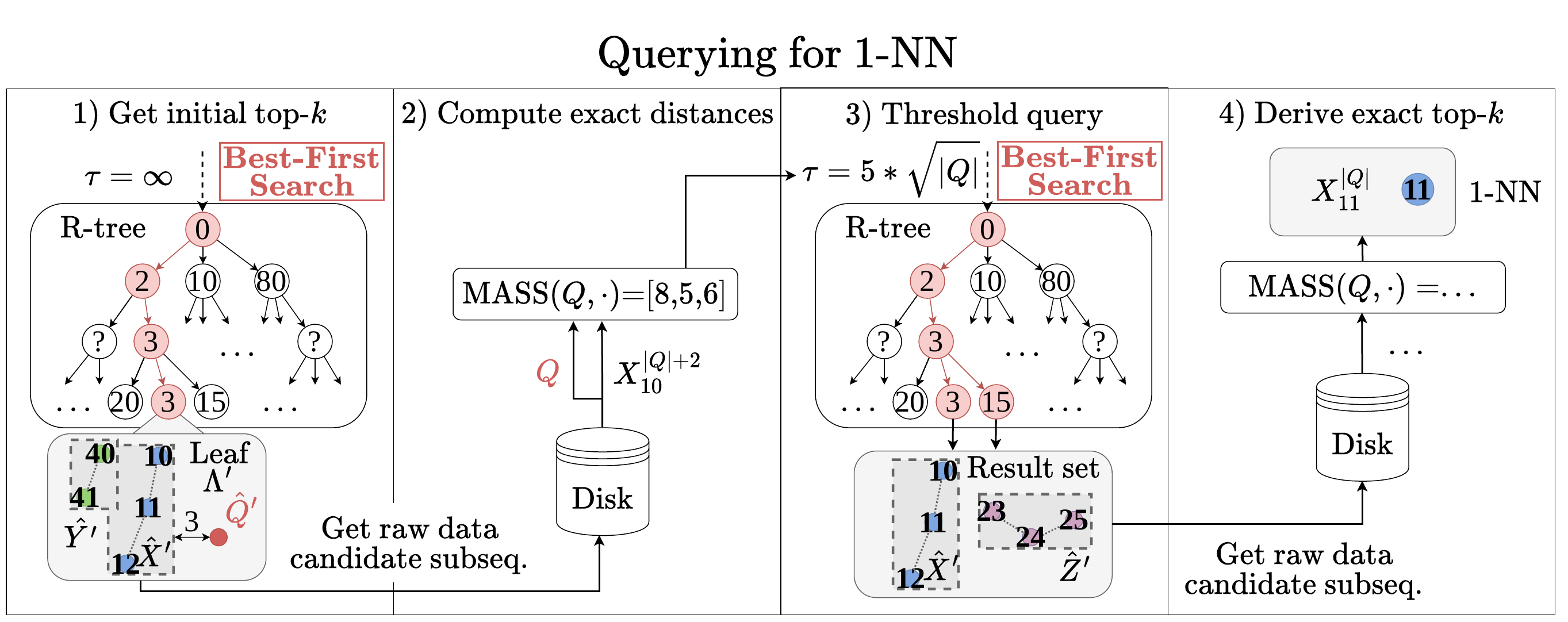}
    \end{subfigure}
    \caption{(a) Cumulative and absolute \% of total distance and energy across DFT coefficients on the temperature channel of a weather dataset; (b) Query execution in \name. Numbers in tree nodes indicate the lower bound distance of the respective MBR to the query.}
    \label{fig:cumdist_and_query}
\end{figure*}

\noindent\textbf{Observation 2:} The contribution of DFT coefficients to the total \emph{distance} between time series is even more significant than their contribution to energy.
Intuitively, this means that the top coefficients are not only larger in scale, but also vary much more between time series. 
This evident from Figure~\ref{fig:cumdist_and_query}a, which shows that the cumulative contribution to distance increases much faster than the contribution to energy, implying that one can already cover $\sim$90\% of the total distance with the top 5 coefficients, even though these only cover roughly $\sim$60\% of the total energy.

\name exploits these observations to create accurate approximations of time series subsequences, adapted to the dataset at hand.
In particular, we summarize the MTS in the dataset as follows.
First, before index construction, we extract a \rev{R2D4}{uniformly random} sample $S$ of subsequences of length $\qlen$ from different time series in the dataset\footnote{\rev{R2D4}{The sample size is a parameter of the algorithm, and is typically set to 100 subsequences. This value was chosen after empirical testing on different datasets, which showed that larger sample sizes did not lead to a noticeable difference in performance.}}, and derive the \emph{average relative distance contribution} (ARDC) of each DFT coefficient and each channel to the Euclidean distance between these subsequences (i.e., the process that generated the plots in Figure~\ref{fig:cumdist_and_query}a).
Then, for each channel in the dataset, the top-$f$ coefficients with the largest ARDC are derived, with $f$ chosen such that the total ARDC of the top-$f$ coefficients is above a certain threshold $d_{\text{target}}$ (e.g., 90\%).
Parameterizing the algorithm on $d_{\text{target}}$ rather than $f$ allows the algorithm to adapt the size of the summarization to the dataset at hand, ensuring a predetermined level of accuracy of the approximations.
Those top-$f$ coefficients are then used to compute the DFT approximations of \emph{all} $\qlen$-length subsequences in the dataset.
Notice that the summation in \rev{R3W2}{Equation~\ref{eq:dft_bound_uts}} can be over any permutation or subset of coefficients; as long as the same coefficients are used for both subsequences the bound still holds.

In the following, with $\avec{T}$ we will denote the matrix of size $c \times f$ that stores the DFT approximations for subsequence $\vec{T}$. 
Also, $\avec{T}'$ will be used to denote the \emph{feature vector} of subsequence $\vec{T}$, which is of length $c*f$ and is constructed by flattening the matrix $\avec{T}$ (step 2 in Figure~\ref{fig:summarization}).
Summarizing a dataset of $n$ time series costs $O(n*m*\qlen*f)$, and leads to a total of $O(n*m)$ feature vectors -- one per subsequence.

\subsection{Indexing of the MTS}\label{sec:indexing}
The next step involves indexing the generated feature vectors in an R-tree (step 3 in Figure~\ref{fig:summarization}).\footnote{
We also experimented with other spatial indexes like KD-trees~\cite{bentley75_kdtree}, but found that the recall stage of our algorithm was the most time-consuming part of our solution, rather than the probing of the index. We thus opted for using an R-tree to enable a more clear comparison with ST-index~\cite{faloutsos94}.
}
Notice, however, that the total number of feature vectors is $O(n*m)$, which can lead to a large and inefficient R-tree when the number of time series $n$ is large, negatively impacting the query performance.

We mitigate this problem as follows. 
First, we build the R-tree in a bottom-up fashion with leaves containing more than one entry (i.e., leaf size $L>1$), to reduce the number of nodes in the tree.
Second, once the leaves are constructed, we revisit them and modify their internal representation, to allow pruning of multiple subsequences (i.e., entries) with a single distance computation.
Particularly, for each leaf that contains two or more entries, we group together all entries that originate from \emph{time-neighbouring subsequences} of the same MTS, i.e., feature vectors corresponding to a series of subsequences that are a single time shift apart (e.g., $\vec{T}_i^{\qlen}$,  $\vec{T}_{i+1}^{\qlen}$,  $\vec{T}_{i+2}^{\qlen}$, $\ldots$).
These groups correspond to a continuous subsequence of the original MTS.  
\rev{R1D2}{For example, assume a leaf $\Lambda$ that includes the following subsequences: $\Lambda=\{\vec{X}_{10}^{\qlen}, \vec{X}_{11}^{\qlen}, \vec{X}_{12}^{\qlen}, \vec{Y}_{40}^{\qlen},\vec{Y}_{41}^{\qlen}\}$, with $X$ and $Y$ being two different MTS.
We compress these five entries into two, representing the relatively larger subsequences $\Lambda'=\{\vec{X}_{10}^{\qlen+2}, \vec{Y}_{3}^{\qlen+1}\}$.
These entries are now represented with Minimal Bounding Rectangles (MBR) rather than with single points in the feature space, with the start and end positions of the continuous subsequence stored in the entry along with the MBR.}

The intuition behind this grouping is that time-neighboring subsequences are typically very similar due to their large overlap~\cite{faloutsos94}.
This means that their feature vectors are also similar, often ending up in the same R-tree leaf.
This allows many subsequences to be represented by compact entries with tight MBRs, which both reduces the number of entries in the R-tree and enables efficient distance computation with MASS.
In our experiments, this compression typically merges 8-50 subsequences into a single entry, depending on the dataset, with a leaf size optimized for query performance.

\rev{R1D2}{
    Consequently, the indexing process of \name is as follows.
    (a) We build an R-tree bottom-up on the feature vectors of all subsequences in the dataset. In our running example with $X$ and $Y$, this step involves creating the leaf node $\Lambda$ along with all other leafs that contain the other subsequences in the dataset. 
    In the recursion of the R-tree, we then create the parent nodes of $\Lambda$, which contain the MBRs of the leafs (and are MBRs themselves), and so on until we reach the root node.
    (b) We then revisit the leaf nodes and group together all time-neighbouring entries (i.e., creating $\Lambda'$) to reduce the index size.
}
\rev{R1D4}{
    After these steps, we end up with an R-tree with each entry storing (a) an MBR covering the indexed feature vectors, (b) the start and end positions of the subsequences in the original MTS, and (c) the MTS identifier.
}
This index is visualized on the right of Figure~\ref{fig:summarization}, where the nodes represent the indexed subsequences of different MTS (differentiated by color), and the edges connect time-neighbouring subsequences of the same MTS.  

\newcut{Lastly, notice that this structure allows new data (either new time series or new values of existing time series) to be indexed in a straightforward manner; one simply (a) extracts the feature vectors of the newly formed subsequences, (b) inserts them into the R-tree following common R-tree insertion procedures, and (c) adds the new entries to time-neighbouring groups if applicable.}

\subsection{Query Execution}\label{sec:query}
Our algorithm probes the index twice.
The first probe is for finding a rough estimate (an upper bound) of the distance to the $k$-th nearest neighbor, whereas the second probe retrieves all entries within this distance, and computes the exact result. \rev{R1D1}{The second probe is necessary to guarantee correctness of the answer}.

Queries are executed as follows.
First, we extract a feature vector $\avec{Q}'$ of the query $\vec{Q}$.
Then, starting with a distance threshold $\tau_k = \infty$, \rev{R3D1}{we perform a best-first search on the R-tree}, to find the $k$ entries (i.e., groups of $\qlen$-length subsequences) with the smallest distance to $\avec{Q}'$.
Since the R-tree relies on the feature vectors to compute all distances, the computed distances are lower bounds (LBs) on the true distances to $\vec{Q}$.
Then, for each retrieved entry, we compute the exact distances between $\vec{Q}$ and the $\qlen$-length subsequences in the entry using MASS, and set the distance threshold $\tau_k$ to the $k$-th smallest distance found in this step.
Finally, we perform a threshold query on the index, using $\tau_k*\sqrt{\qlen}$ as the distance threshold, and compute the exact distances for these entries with MASS to obtain the final \knn.\footnote{\rev{R1D4}{Note that MASS is not executed on the feature vectors, but rather on the original subsequences, to ensure that the distances are exact. This data is retrieved by chasing the pointers stored in the R-tree entries to the original MTS, residing on disk.}}
\rev{R1D4}{
These results are then returned to the user along with the corresponding timestamps of the subsequences, which are derived from the offset and the timestamp of the MTS.}
Since the query $\vec{Q}$ does not necessarily contain all channels (i.e., $c_{\vec{Q}} \subseteq \{1,\dots,c\}$), the R-tree is only queried on the dimensions of the feature space that correspond to the channels in $\vec{Q}$. 
This is natively handled by the R-tree algorithm~\cite{samet05}.

\rev{R1D2}{
    To illustrate this process, consider a query $\vec{Q}$ for which the 1-NN is the subsequence $\vec{X}_{11}^{\qlen}$ in leaf $\Lambda'$ from the previous section.
    The algorithm will return the entry $\vec{X}_{10}^{\qlen+2}$ in the first probe, as it has the smallest lower bounding distance to $\vec{Q}$ (Step 1 in Figure~\ref{fig:cumdist_and_query}b).
    It then computes MASS$(\vec{Q},\vec{X}_{10}^{\qlen+2})$ on the original subsequences, which outputs the distances $[8,5,6]$, and adds $\vec{X}_{11}^{\qlen}$ to the running 1-NN set as it corresponds to the smallest distance of 5 (Step 2).
    Accordingly, the threshold is set to $5*\sqrt{\qlen}$ and the algorithm performs the second probe, retrieving the entries $\vec{Z}_{23}^{\qlen+2}$ and $\vec{X}_{10}^{\qlen+2}$ (Step 3).
    After computing the exact distances for these entries with MASS, the algorithm returns $\vec{X}_{11}^{\qlen}$ as the 1-NN to $\vec{Q}$ (Step 4).
}

We now prove that the query process described above is \rev{R1D1}{guaranteed to return the complete and correct \knn to the query $\vec{Q}$}.
Specifically, we aim to prove the following lemma:
\begin{lemma}\label{lem:correctness}
    \rev{R1D1}{Any indexed subsequence $\vec{T}$ of length $\qlen$ that is part of the \knn of $\vec{Q}$ is guaranteed to be in the set of subsequences returned by \name.}
\end{lemma}
\begin{proof}
    The first probe of the index returns $k$ subsequences of length $s \geq \qlen$, for which the exact distances with the query are computed. As these $k$ subsequences collectively contain at least $k$ subsequences of length $\qlen$, by setting $\tau_k$ as the k'th smallest of these distances we get an upper bound on the distance of the true $k$-th nearest neighbor.
    In the second probe, we perform a range query with threshold  $\tau_k*\sqrt{\qlen}$ on the R-tree index, which contains the feature vectors.
    Based on Equation~\ref{eq:dft_bound_uts}, the distance between any two points in the feature space is a lower bound on the true distance of the subsequences they represent.
    Specifically, the distance between a $\qlen$-length subsequence $\vec{T}$ and the query $\vec{Q}$ is lower-bounded by the distance between their feature vectors $\avec{T}'$ and $\avec{Q}'$ as:
    \begin{equation}\label{eq:dft_bound_mts}
        d(\vec{T},\vec{Q}) \geq \frac{1}{\sqrt{\qlen}} \sqrt{\sum_{i\in c_{\vec{Q}}} d(\avec{T}_i,\avec{Q}_i)^2} = \frac{d(\avec{T}'_{c_{\vec{Q}}},\avec{Q}'_{c_{\vec{Q}}})}{\sqrt{\qlen}}
    \end{equation}
    where $\avec{T}'_{c_{\vec{Q}}}$ and $\avec{Q}'_{c_{\vec{Q}}}$ correspond to the feature vectors of $\vec{T}$ and $\vec{Q}$ limited to the channels of $\vec{Q}$.
    Now, assume that an indexed subsequence $\vec{T}$ has distance with $\vec{Q}$ less than $\tau_k$. 
    Then, by Equation~\ref{eq:dft_bound_mts} we know that $\frac{d(\avec{T}_{c_{\vec{Q}}},\avec{Q}'_{c_{\vec{Q}}})}{\sqrt{\qlen}} \leq d(\vec{T},\vec{Q}) \leq \tau_k$.
    Therefore, by querying an R-tree a threshold $\tau_k * \sqrt{\qlen}$, $\vec{T}$ is guaranteed to be included in the returned set.\footnote{Note that the $\sqrt{\qlen}$ is simply a scaling factor that needs to be accounted when transforming the distances from the feature space to the time domain. It does not weaken the bounds.}
\end{proof}
\rev{R1D1}{
The lemma proves that \name is complete by guaranteeing that all subsequences in the \knn are found. 
Since exact distances are computed for all candidates using MASS (which is exact), and the \knn is derived from these distances, the algorithm is also correct, never returning false positives.
}

\subsection{Optimizations}\label{sec:optimizations}
We now present three optimizations of \name aimed at improving the pruning power of the index and the efficiency of the search process.
\rev{R1D1}{These optimizations do not affect the correctness of the algorithm or the final result -- \emph{the algorithm still remains exact.}}

\medskip\textit{Tightening the DFT bounds.}
As discussed in Section~\ref{sec:summarization}, in most real-world datasets, around 60\%-80\% of the distance between time series is covered by the top 2-5 DFT coefficients.
Formally, this concept is expressed as:
\begin{equation}\label{eq:dft_bound_twoparts}
d^2(\vec{T},\vec{Q}) = (\underbrace{d^2(\avec{T}_f,\avec{Q}_f)}_{\sim 80\%} + \underbrace{d^2(\avec{T}_{f+},\avec{Q}_{f+})}_{\sim 20\%}) / \qlen
\end{equation} 
where $\avec{T}_f$ denotes the DFT approximation of an MTS $\vec{T}$ using the top-$f$ coefficients, and $\avec{T}_{f+}$ denotes the set of $c$ vectors composed of the remaining coefficients.
In the DFT bound of Equation~\ref{eq:dft_bound_mts}, we effectively throw away the second term, which leads to a lower bound of the true distance.
This bound can be computed very efficiently (in $O(f)$), but it may also lead to a weaker pruning in the R-tree, and to unnecessary distance computations for subsequences that are outside the final \knn.

We improve this lower bound such that it approaches the true distance more closely, by introducing a correction term at the distance calculation that approximates (again by lower-bounding) the distance over the remaining $(\qlen-f)$ coefficients, \emph{without actually having to compute these coefficients}. Computing this correction term costs $O(1)$ at query time for each probed R-tree node, and improves the pruning power of the index. 

Namely, during index construction, after computing the DFT approximation for each subsequence $\vec{T}$, we also derive the part of $\vec{T}$ that is not covered by the top-$f$ coefficients, called its \emph{remainder} $\vec{T}_{f+} = [\vec{T}_1 - \textsc{IDFT}(\avec{T}_{f,1}), \dots, \vec{T}_c - \textsc{IDFT}(\avec{T}_{f,c})]$.
This remainder can be computed solely based on the top-$f$ coefficients, and does not require computing all other $(\qlen-f)$ coefficients.
The key observation here is that the distance between the remainders of two time series captures the remaining $\sim$20\% of their distance as shown in Equation~\ref{eq:dft_bound_twoparts}.
Then, Equation~\ref{eq:dft_bound_twoparts} can be rewritten to utilize the remainders:
\begin{equation} \label{eqn:withoutcorrection}
    d^2(\vec{T},\vec{Q}) = \frac{d^2(\avec{T}_f,\avec{Q}_f)}{\qlen} + d^2(\vec{T}_{f+},\vec{Q}_{f+})
\end{equation}
The remainders are of length  $\qlen$, meaning that the cost is as much as computing the full distance on the original data (i.e., $O(|c_{\vec{Q}}|*\qlen)$). To avoid this cost at query time, we precompute the distances of the remainders for each subsequence to a small fixed set of \emph{pivot points} at index time, and store these distances in a map.
Then, during query execution, we do the same for $\vec{Q}_{f+}$, and use the reverse triangle inequality to lower-bound the distance between $\vec{Q}_{f+}$ and the remainders of the indexed subsequences through their distances to pivots.
We refer to this latter bound as the \emph{correction term}. We apply the correction term to Equation~\ref{eqn:withoutcorrection} as follows:
\begin{eqnarray}\label{eqn:withpivot}
0 & \leq (|d(\vec{T}_{f+},\vec{P})-d(\vec{Q}_{f+},\vec{P})|)^2 \leq  d^2(\vec{T}_{f+},\vec{Q}_{f+}) \nonumber \Rightarrow\\
     d^2(\vec{T},\vec{Q}) & \geq \underbrace{d^2(\avec{T}_f,\avec{Q}_f) / \qlen}_{\text{DFT distance}} + \underbrace{(|d(\vec{T}_{f+},\vec{P})-d(\vec{Q}_{f+},\vec{P})|)^2}_{\text{Correction term}}
\end{eqnarray}
where $\vec{P}$ is a pivot.
At index construction time, we derive $k$ pivots by running $k$-means clustering on a sample of the remainders of subsequences in the dataset. Then, during query execution, the algorithm  detects the closest pivot to the query's remainder, and uses this pivot $\vec{P}$ to compute the bound, according to Equation~\ref{eqn:withpivot}.
In Section~\ref{sec:evaluation} we show that this optimization leads to a 2x speedup in query execution.

\begin{figure}[t]
    \centering
    \includegraphics[width=\linewidth]{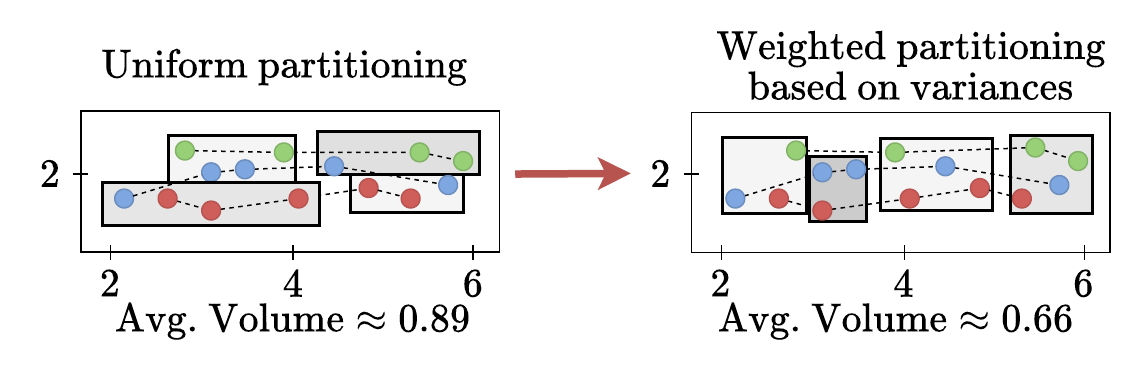}
    \caption{Different partitioning strategies for a 2-dimensional feature space; the STR algorithm (left) and the proposed weighted partitioning (right), that leads to smaller MBRs and to tighter bounds. \cut{The diagrams represent the feature space, and the points of different colors represent the DFT approximations for subsequences of three different MTS (one color per MTS). For illustration, time-neighboring subsequences are connected with a line.}}
    \label{fig:partitioning}
\end{figure}

\paragraph{Tightening the MBRs.}
Pruning efficiency of the R-tree can be further enhanced by reducing the volume of the MBRs in the R-tree nodes.
The R-tree is constructed in a bottom-up fashion using the STR algorithm~\cite{leutenegger97}, briefly presented in Section~\ref{sec:rtree}.
STR's approach of splitting the entries into an equal number of partitions per dimension works well when the data is uniformly distributed at all dimensions, but it can lead to sub-optimal partitioning when the data is skewed or concentrated in certain areas of the space.
To illustrate the reason, consider a simple example where the R-tree is used to index a two-dimensional space, with the values in the first dimension ranging from 2 to 6, and in the second from 1.5 to 2.5 (see Figure~\ref{fig:partitioning}). 
The first dimension will likely be the main contributor of the Euclidean distances between different entries. 
Therefore, if we devote more splits/partitions on the first dimension while constructing the R-tree, the resulting MBRs will be tighter in this dimension, leading to tighter lower distance bounds and to a more aggressive pruning. 
Figure~\ref{fig:partitioning} illustrates this idea.  

DFT approximations, in particular, are prone to have such a heavy concentration of variance in the first few dimensions, as discussed in Section~\ref{sec:summarization}.
We leverage this observation by weighing the dimensions of the R-tree (i.e., the DFT coefficients across channels) based on the variance of their values, which is a good proxy for the contribution of the dimension to the distance between points. 
Namely, before constructing the R-tree, \rev{R2D5}{we use the sample $S$ of subsequences extracted during the summarization step (Section~\ref{sec:summarization}) to estimate the variance $\hat{\sigma}_i^2$ of the distribution of feature vectors across each dimension $i$ of the feature space.}
Next, we compute $\omega_i$ (the weight of dimension $i$) through  softmax normalization of the variances~\cite{softmax}.
Then, given a desired leaf size $L$, the number of splits $p_i$ for a dimension $i$ is determined  by $p_i = \lceil (N/L)^{\omega_i}\rceil$, with $N$ being the total number of entries to index in the whole dataset.
This definition of $p_i$ ensures that the desired leaf size $L$ is reached as $\prod_{i=1}^{c*f*2} \lceil (N/L)^{\hat{\omega_i}} \rceil \approx \frac{N}{L}$.
This way, the algorithm effectively re-distributes the number of partitions across the dimensions, such that it is proportional to the contribution of each dimension to the distance.
In Section~\ref{sec:evaluation} we will show that this optimization leads to a 2-4x speedup in query execution, due to the more aggressive pruning of the index.


\rev{R3D1}{
\paragraph{Distance browsing.}
    Notice that, as the threshold $\tau_k$ used in the second probe is set to the $k$-th smallest \emph{exact} distance of the subsequences returned by the first probe, the second probe is guaranteed to return a superset of the entries returned by the first probe.
    To avoid redundant lower bound distance computations, we preserve the priority queue of entries used in the best-first searches between the two probes, so that the second probe continues from where the first one left off.
    This optimization is commonly used for querying spatial indices, and  is typically referred to as \emph{distance browsing}~\cite{hjaltason99}.
}

%% file: sections/4.baselines.tex
\section{EXTENDING EXISTING ALGORITHMS TO THE MULTIVARIATE CASE}
\label{sec:baseline}

\begin{algorithm}[bt]
    \caption{\textsc{UTSBaseline}($\mathcal{I},\q,k$)}
    \label{alg:baseline}
    \SetKwInOut{Input}{Input}
    \SetKwInOut{Output}{Output}
    \Input{A set of channel-level indices $\mathcal{I}$, a query MTS $\q$, a result set size $k$.}
    \Output{The top-$k$ subsequences in $\mathcal{T}$ with the lowest distance to $\q$.}
    $\hat{\mathcal{R}} \leftarrow \{ \}$ \\
    \For(\tcp*[f]{Iterate over indices}){$i \in c_{\q}$}{
        $\hat{\mathcal{R}} \leftarrow \hat{\mathcal{R}} \cup \textsc{Query}(\mathcal{I}_i, \q_{i}, k)$ \tcp*[f]{Query index} \\
    }
    $\hat{\mathcal{R}} \leftarrow \textsc{ExhaustiveTopK}(\hat{\mathcal{R}}, \q)$ \tcp*[f]{Compute est. top-$k$} \\
    \For(\tcp*[f]{Set thresholds}){$i \in c_{\q}$}{
        $\tau_i \leftarrow \underset{S\in \hat{\mathcal{R}}}{\max}~ d^2(S_{i},\q_{i})$\;
    }
    $\mathcal{R} \leftarrow \{ \}$\;
    \For(\tcp*[f]{Re-query}){$i \in c_{\q}$}{
        $\mathcal{R}_i \leftarrow \textsc{Query}(\mathcal{I}_i, \q, \tau_i)$\;
        $\mathcal{R} \leftarrow \mathcal{R} \cup \{\vec{T} | \vec{T} \in \mathcal{R}_i \wedge d^2(\q_{i},\vec{T}_{i}) \leq \tau_i\}$ \\
    }
    \Return $\textsc{ExhTopK}(\mathcal{R}, \q)$\;
\end{algorithm}

As mentioned in Section~\ref{sec:related}, there exist several efficient algorithms that address similarity search for UTS. 
Therefore, a natural question is whether these algorithms can be extended to query MTS, and how they would perform in such a setting.
We will now present a unified extension -- a wrapper algorithm -- that can be used to extend any UTS search algorithm to work with MTS.
We also note that this alternative design approach comes with several deficiencies compared to \name, such as low pruning power when the channels are uncorrelated, as shown in Section~\ref{sec:evaluation}.
Still, it is a useful approach to enable comparison of \name with out-of-the-box extensions of UTS algorithms such as DSTree~\cite{dstree}, ST-index~\cite{faloutsos94}, and KV-match~\cite{kvmatch} on MTS data.

The general approach is based on the well-known Threshold algorithm~\cite{fagin01}, which can be used to derive a global top-$k$ from multiple sorted lists with different attribute values of the same objects, given a monotonic aggregation function to compute the target value upon which the top-$k$ is based.
\newcut{In our case, the objects are the MTS subsequences in the dataset, the values are the squared Euclidean distances to the query on a certain channel, and the aggregation function is the rooted sum of these distances (cf. Equation REF).} 
The approach works as follows (cf., Alg.~\ref{alg:baseline}):
(a) We initialize one index (e.g., one DSTree or one ST-index) per channel.
(b) At query time, we first obtain an initial top-$k$ estimate $\hat{R}_i$ for each channel $i \in c_{\q}$ by querying the corresponding index (lines 2-3).
(c) For each subsequence in the union of the local estimates, we compute the full distance to the query (i.e., using all query channels), constructing an intermediate global top-$k$ $\hat{R}$ (line 4).
(d) We set a distance threshold $\tau_i$ for each channel $i \in c_{\q}$ to the largest \emph{univariate} Euclidean distance in $\hat{R}$ on that channel (lines 5-6).
(e) Finally, we re-query the channel-level indices with their respective thresholds $\tau_i$. 
We compute the full distances to the query for the results, update the global top-k accordingly, and return it as a final result (lines 7-11).
Since any subsequence that belongs to the global top-$k$ must have a distance at most $\tau_i$ on at least one channel $i \in c_{\q}$, \rev{R1D1}{the result of this algorithm is guaranteed to be correct.}
Lastly, MASS is used to speed up the exact distance computations (lines 4 \& 10).

%% file: sections/5.evaluation.tex
\section{EVALUATION}\label{sec:evaluation}
The purpose of our experiments was threefold: (a) to assess the scalability and efficiency of \name under various query configurations, (b) to compare \name to other methods, and, 
(c) to test the effectiveness of the optimizations proposed in Section~\ref{sec:optimizations}.
\rev{R1D1,R1D12}{Since our method and all baselines are exact, our evaluation only focuses on efficiency; experiments on accuracy would always yield 100\%.}

\smallskip\textit{Compared methods.}
    \rev{R1D7}{
    First of all, we compare \name to \textbf{MULIS\-SE} \cite{mulisse}, applying it to fixed-length search by restricting the supported range of query lengths to $\qlen$.
    This way, MULISSE also exploits the knowledge of a preknown query length, which slightly improved its performance without affecting the quality of the results.
    As no other method natively supports \knn MTS subsequence search, our other baselines comprise of SOTA methods for UTS search extended through Algorithm~\ref{alg:baseline}.\footnote{We add an asterisk to their names to differentiate them from their original versions.}
}
(a) \textbf{\stindex}~\cite{faloutsos94}, a well-known index for subsequence search~\cite{moon02,zhu12},
(b) \textbf{\kvmatch}~\cite{kvmatch}, a SOTA index for subsequence search on single long series, extended with per-MTS indices, and
(c) \textbf{\dstree}~\cite{dstree}, an index for whole-matching~\cite{hydra_exact}, adapted by indexing all subsequences.
We also include two sequential scan baselines:
(a) \textbf{MASS}~\cite{mueen20} applied to all MTS, and
(b) \textbf{Brute-force} exhaustive comparison.
Detailed descriptions of the compared methods can be found in Section~\ref{sec:related}.

\smallskip\textit{Hardware and implementations.}  
All experiments were executed on a server equipped with a 64-core 2.4 GHz AMD Genoa 9654 processor and 128 GB of RAM. 
\rev{R1D7}{The code for MULISSE (C++) was provided by the authors of~\cite{mulisse}, and was executed directly with the original parameters as described in the paper.}
The code for \dstree (Java-11) was provided by the authors of~\cite{dstree}, and was invoked as a subroutine in Alg.~\ref{alg:baseline}.
Due to the lack of publicly available code, all other algorithms were implemented from scratch in Java-11.\footnote{Our code is available at \label{github}\vldbavailabilityurl}
\rev{R2D6}{For a fair comparison, all methods were implemented to operate fully in \emph{main memory} and run in \emph{single-threaded}.}

\smallskip\textit{Datasets.}
We used 32 real-world publicly available datasets from different domains, and a set of synthetic datasets:
(a) \textbf{Stocks.} Daily volumes, opening, closing, high, and low-prices of 28678 stocks over the period Jan. 2, 1987 to Feb. 26, 2021. 
(b) \textbf{Weather.} Segment of the ISD weather dataset~\cite{noaaISD} containing hourly readings of wind speed, sea level pressure, atmospheric temperature, dew point temperature
of 13545 sensors taken between Jan. 1, 2020 and Dec. 31, 2021.
\rev{R2D7}{
    (c) \textbf{Wind.} Sensor output of an active wind turbine sampled every 2 seconds for 10 days~\cite{wind_dataset}, resulting in a single MTS with 432,000 observations and 10 channels covering power output, rotor speed, wind speed, and other wind-related variables.
}
(d) \textbf{Synthetic.} Random walk datasets with 1600 MTS, of 64 channels and 4096 observations each. Following~\cite{hydra_approx,hydra_exact}, these datasets were generated by drawing random numbers from a normal distribution with a mean of 0 and a standard deviation sampled uniformly at random from the interval $[0,\dots,10]$, with starting points sampled uniformly at random from the interval $[0,\dots,100]$. 
(e) \textbf{UEA archive.} Popular benchmark archive for MTS classification consisting of 30 labeled real-world datasets from different domains~\cite{bagnall18}, with varying numbers of MTS, channels, and observations.
Train and test splits for each dataset were merged, to form a single dataset. 

We evaluated both raw and normalized subsequences, focusing on raw results and highlighting normalized results only when they provide additional insights.
The Stocks dataset serves as our primary benchmark due to its size, with other datasets discussed in designated experiments or when they reveal notable patterns.
Table~\ref{tab:datasets} summarizes the default query confirms for each dataset. 
We include more details about the datasets in our code repository.\footref{github}


\begin{table}[tb]
    \centering
    \caption{Default query configurations.}
    \label{tab:datasets}
    \resizebox{\columnwidth}{!}{
    \begin{tabular}{l|lllll}
        \hline
         & \textbf{Stocks} & \textbf{Weather} & \textbf{Synthetic} & \textbf{Wind} & \textbf{UEA (30x)} \\\hline
        $n$ & 2000 & 1300 & 1600 & 1 & all, see~\cite{bagnall18}\\ 
        Avg. $m$ & 5590 & 8692 &  4096 & 432,000 & see~\cite{bagnall18} \\ 
        $\qlen$  & 730 (2 y.) & 1488 (2 mth.)  & 1024 & 1800 (1 h.) & 20\%\\ 
        \#Channels & 5 & 4 & 64 & 10 & see~\cite{bagnall18} \\ 
        \hline
    \end{tabular}
    }
\end{table}

\smallskip\textit{Queries and evaluation metrics.}
As per standard practice~\cite{hydra_exact}, we generated the query workloads for each dataset by randomly selecting $\qlen$-length subsequences from the dataset and adding Gaussian noise with standard deviation $0.1*\sigma_{Q_i}$ on all channels.
As part of sensitivity analysis experiments, we also investigate varying noise levels and query generation from out-of-dataset subsequences in Section~\ref{sec:evaluation:querynoise}.
Unless otherwise stated, we query on all channels of the dataset.
When querying on a subset of channels (as in Section~\ref{sec:evaluation:query_channels}), the query channels are selected uniformly at random from all available channels. 
For each algorithm, we measured: (a) initialization time, e.g., index construction index, (b) query execution time, and, (c) size of the index and/or auxiliary data.
We set a timeout of 12 hours for total execution time, and report median results over 10 runs with 100 queries each.


\subsection{Tuning the Algorithms}\label{sec:parameterization}

\subsubsection{Number of DFT coefficients for \name and \stindex}
Recall from Section~\ref{sec:summarization} that the number of DFT coefficients $f$ is derived based on the distance $d_{\text{target}}$ covered by the top-$f$ coefficients, rather than parameterizing the number of coefficients directly.
For both algorithms, $d_{\text{target}}$ was tuned through a grid search over the values [20\%, 40\%, \ldots 100\%], across all datasets and with queries on both raw and normalized subsequences.
The results showed that a coverage of 60\% was the most robust choice for both algorithms across the datasets, resulting in a query time at most 20\% larger than the optimal choice for each dataset.
Therefore, $d_{\text{target}}$ was set to 60\% for the following experiments.

\begin{table}[t]
    \caption{Average query time (ms) of \name for different leaf size values. Leaf size is expressed as a percentage of the total number of subsequences in each dataset (raw datasets).}
    \label{tab:leafsize}
    \small
    \begin{tabular}{l|lrrrrrrrr}
        \hline
        Leaf size & Insectwingbeat & Stocks & Synthetic & Weather \\
        \hline
        0.0001 \% & \textbf{8.61} & 8.50 & 2.98 & 4.67 \\
        0.01 \% & \textbf{8.61} & 8.37 & \textbf{2.97} & 4.64 \\
        0.05* \% & 9.01 & \textbf{8.16} & \textbf{2.98} & \textbf{4.61} \\
        0.1 \% & 10.79 & 8.36 & 3.01 & 4.69 \\
        1 \% & 22.70 & 11.83 & 3.62 & 5.06 \\
        100 \% & 38.12 & 158.10 & 4.59 & 100.87 \\
        \hline
        \end{tabular}
\end{table}

\subsubsection{Leaf size}\label{sec:evaluation:leaf_size}
\dstree, \stindex, and \name support tuning of the leaf size $L$, i.e., the \emph{maximum} number of entries per leaf. 
Similar to the number of DFT coefficients, we tuned the leaf size for query time through a grid search over values ranging from 0.0001\% to 100\% (i.e., no constraints on the leaf size) of the total number of entries to index.\footnote{\stindex's leaf size was tuned after tuning the number of DFT coefficients, using a leaf size of 0.1\% in the first step.}
Looking at the results in Table~\ref{tab:leafsize}, we see that a leaf size of 0.05\% provides consistently good performance across all datasets for \name.
Therefore, this leaf size was chosen for all experiments.
For \stindex, the optimal leaf size was also 0.05\%.
For  \dstree, the optimal leaf size was found to be 10\%. This contradicts the results of Echihabi et al.~\cite{hydra_exact}, who showed that the optimal leaf size for query time was 0.9\% for time series of size 256. Both 0.9\% and 10\% had a similar query time in our experiments, but the 10\% choice had a substantially lower initialization cost.

\subsubsection{\kvmatch segment size}\label{sec:kvmatch:parameters}
The segment size parameter in \kvmatch controls the number of piecewise means and variances used to index each subsequence.
We found the optimal segment size to be 1 for all datasets.
For raw subsequences, this reduces the index to a single lookup table per channel and time series.
For normalized subsequences, all entries end up in the same bucket due to zero mean and unit variance, effectively reducing \kvmatch to MASS with overhead.

\begin{figure*}[tb]
    \begin{subfigure}[t]{0.19\textwidth}
        \centering
        \includegraphics[width=\textwidth]{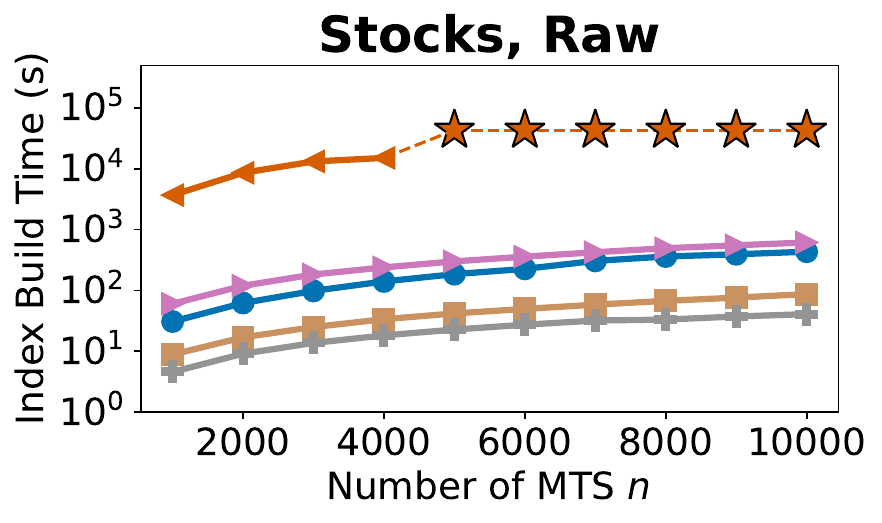}
    \end{subfigure}
    \begin{subfigure}[t]{0.19\textwidth}
        \centering
        \includegraphics[width=\textwidth]{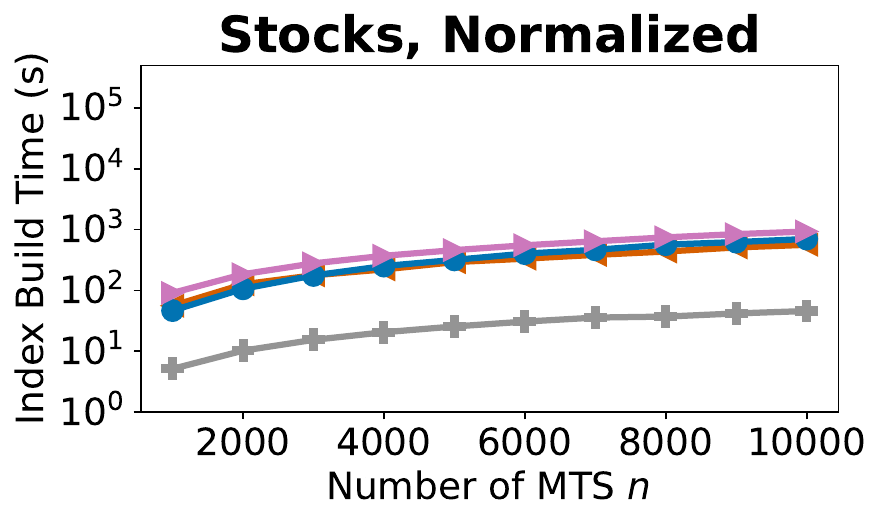}
    \end{subfigure}
    \begin{subfigure}[t]{0.19\textwidth}
        \centering
        \includegraphics[width=\textwidth]{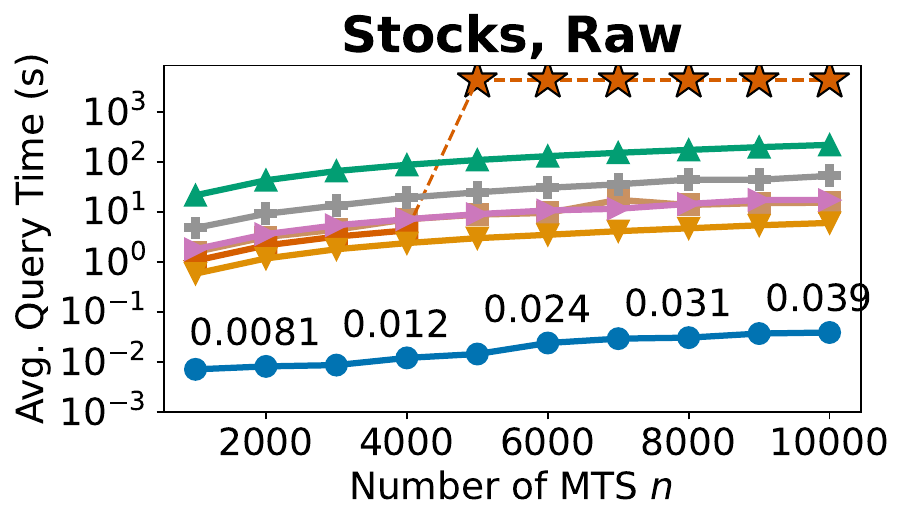}
    \end{subfigure}
    \begin{subfigure}[t]{0.19\textwidth}
        \centering
        \includegraphics[width=\textwidth]{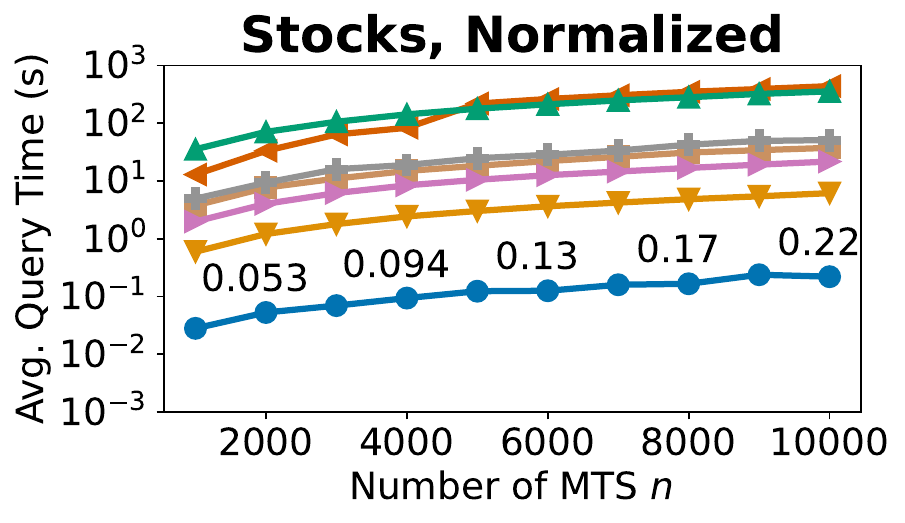}
    \end{subfigure}
    \begin{subfigure}[t]{0.19\textwidth}
        \centering
        \includegraphics[width=\textwidth]{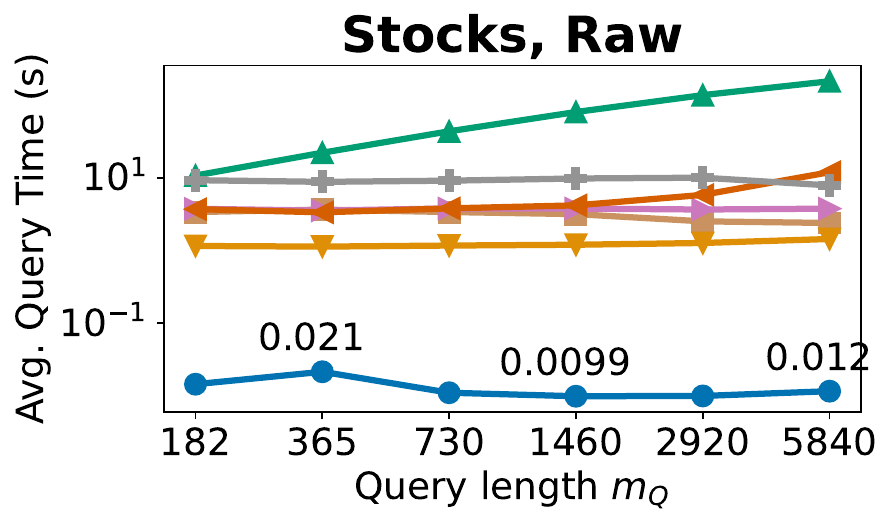}
    \end{subfigure}
    \begin{minipage}[b]{\textwidth}
        \centering
       \includegraphics[width=.7\textwidth]{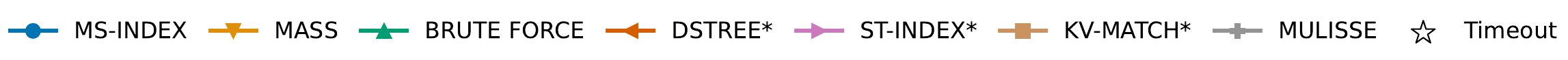}
    \end{minipage}
    \caption{Scalability over number of MTS $n$ on \emph{Stocks}. (a) Initialization time for raw subsequences, and (b) normalized subsequences; (c) Query time for raw subsequences, and (d) normalized subsequences. (e) Query time for raw subsequences with varying query length. All y axes are presented in log scale.}
    \label{fig:row1}
    \centering
    \includegraphics[width=.8\textwidth]{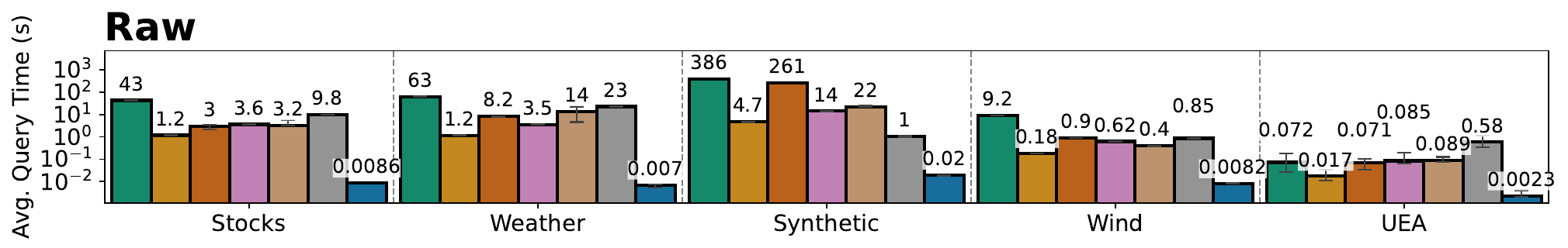}
    \begin{minipage}[b]{\textwidth}
        \centering
       \includegraphics[width=.8\textwidth]{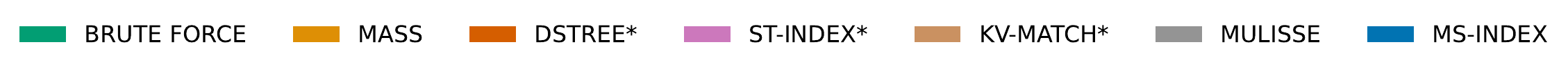}
    \end{minipage}
    \caption{\rev{R2D7}{Query time of algorithms on different datasets. All y axes are presented in log scale.}}
    \label{fig:datasets}
\end{figure*}

\begin{table}[tb]
    \centering
    \caption{Index size (MB) \rev{R2D8}{and the corresponding percentage of the dataset size (MB)} of different algorithms across different Stocks dataset sizes.}
    \label{tab:index_size}
    \setlength{\tabcolsep}{3pt} 
    \resizebox{\columnwidth}{!}{
    \begin{tabular}{l|rrrr|rrrr}
        \hline
         & \multicolumn{4}{c|}{\textbf{Raw}} & \multicolumn{4}{c}{\textbf{Normalized}} \\
        $n$ & \textbf{1000} & \textbf{2000} & \textbf{3000} & \textbf{ \%} & \textbf{1000} & \textbf{2000} & \textbf{3000} & \textbf{\%} \\
        \hline
        Dataset size (MB) & 210 & 430 & 641 & 100\% & 210 & 430 & 641 & 100\% \\ \hline
        MULISSE & 14 & 15 & 16 & 5\% & 16 & 18 & 20 & 5\% \\
        \kvmatch & 259 & 536 & 793 & 124\% & 10 & 20 & 30 & 5\% \\
        MASS & 315 & 646 & 964 & 150\% & 315 & 646 & 964 & 150\% \\
        \textbf{\name} & 968 & 1923 & 2861 & 452\% & 1094 & 2248 & 3348 & 522\% \\
        \dstree & 1491 & 3230 & 4773 & 731\% & 620 & 1276 & 1900 & 296\% \\
        \stindex & 2904 & 5771 & 8584 & 1355\% & 3282 & 6746 & 10045 & 1567\% \\
        \hline
    \end{tabular}
    }
\end{table}

\subsection{Evaluation Results}
\subsubsection{Initialization time}\label{sec:evaluation:index}
Figures~\ref{fig:row1}a-b show the initialization time of all index-based algorithms on the Stocks dataset as the number of MTS ($n$) increases. 
All methods scale linearly with $n$, which is expected since initialization of these methods requires iterating over all time series subsequences.
\name and \stindex have comparable initialization times, both primarily spent on computing DFT approximations. 
The remainder of the initialization cost relates to index construction which for  \stindex involves building channel-level trees (scaling linearly with the number of channels).
In contrast, \name builds a single tree for all channels, making it 2-3 times faster for this phase.
\dstree shows comparable initialization time to other methods for normalized subsequences, but is two orders of magnitude slower for raw subsequences. 
This difference occurs because raw subsequences with varying scales and means cause \dstree to create deep, unbalanced trees with costly node splitting operations, while normalized subsequences lead to more balanced trees based on pattern differences.
\rev{R1D7}{
    \kvmatch and MULISSE have the lowest initialization time. 
    This is because \kvmatch only requires a small lookup table for each time series, and MULISSE is originally build and parameterized for variable-length subsequence search, which adds an additional degree of freedom to the problem and makes indexing performance a key priority.
}

\subsubsection{Size of the data structures}\label{sec:evaluation:index_size}
Table~\ref{tab:index_size} presents the memory requirements of each algorithm for the Stocks dataset \rev{R2D8}{both in MB and as a percentage of the dataset size}.
\rev{R1D7}{All methods scale linearly with $n$, with MULISSE being the most space-efficient (storing only $n$ envelopes in a shallow iSAX tree), followed by \kvmatch, MASS, \name, \dstree, and \stindex.}
\name and \stindex have a worst-case space complexity of $O(n*(m_{\max} - \qlen) * f * c)$, with \stindex requiring approximately three times more space than \name due to its channel-level indices.
\rev{R2D8}{
    Reflecting on the relative memory requirements, we see that many methods require more than 100\% of the dataset size to store their indices.
    While this may sound counterintuitive, recall that these structures index the space of \emph{subsequences}, which is significantly larger than the dataset, given that subsequences overlap.
    Indicatively, while a $n=1000$ subset of Stocks is 210 MB large, storing all subsequences of length $\qlen=730$ would require 133 GB of memory, implying that \name already compresses this space by a factor of 138.
    Still, the memory footprint of subsequence indices may be a concern for some applications, and is a commonly acknowledged limitation of these methods~\cite{linardi18_ulisse,ucr_suite}.
}

\subsubsection{Query time}\label{sec:evaluation:query_time}
Figures~\ref{fig:row1}c-d present the query execution time for all methods, when executing queries on subsets of the time series in the Stocks dataset.
We see that \name significantly outperforms all other methods.
In fact, it outperforms its closest competitor, MASS, by over two orders of magnitude, and the other methods by over three orders of magnitude.
These results can be explained by the fact that \name adds pruning power to the MASS algorithm, which already offers a big performance boost over the Brute-force algorithm.
A deeper investigation revealed that \name reaches a median pruning effectiveness of 99\% across all datasets, i.e., 99\% of the subsequences are already pruned from the index and do not need to be compared with MASS.
\rev{R1D7}{
    While \stindex, \kvmatch, and \dstree also act as a pre-filter on top of MASS, their pruning power is significantly lower compared to \name, with pruning effectiveness of 52\%, 65\%, and 46\% respectively on average on the Stocks dataset with raw subsequences.
    This is because the pruning power of these methods depends on the relative ordering of the subsequences to the partial (per-channel) results.
}
When this relative ordering differed significantly over the different query channels, it resulted in relatively high thresholds for the per-channel distances (set in Lines 5-7 of Alg.~\ref{alg:baseline}), and therefore to a high number of MASS-based comparisons.
\name avoids this issue by querying all query channels simultaneously.
\rev{R1D7}{
    Lastly, MULISSE performs second-slowest (after Brute-force) with only 9\% pruning rate.
    This shows that while its extremely memory-efficient design achieves the smallest index size among all methods, this comes at the cost of loose distance bounds that result in many exhaustive comparisons.
}

\subsubsection{Query length}\label{sec:evaluation_query_length}
Since query length must be specified before index construction, we evaluate its impact on performance.
Figure~\ref{fig:row1}e shows that all algorithms except Brute-force are invariant to query length, as they use MASS which scales with subsequence length ($O(m \log m)$ with $m \geq \qlen$) rather than query length.
This suggests \name's performance is independent of $\qlen$, allowing users to choose lengths that best fit their use case.


\subsubsection{Different datasets}\label{sec:evaluation:datasets}
Figure~\ref{fig:datasets} presents the query execution performance across all datasets using default query parameters from Table~\ref{tab:datasets}.\footnote{\rev{R2D7}{The results on normalized subsequences are omitted as they show similar patterns.}}
For the 30 UEA datasets, we present average results as individual dataset results showed similar patterns.
\name consistently outperforms competitors by 2-3 orders of magnitude on all datasets, with the exception of the UEA datasets where \name's advantage is reduced to 5-7 times compared to MASS.
This is because UEA datasets capture short, single events (e.g., a person lifting their arm, or a duck producing a sound) rather than continuous long series with multiple events.
Subsequences in such datasets are typically more similar to each other than those in the Stocks or Weather datasets.
This creates more challenging queries for index-based methods, further discussed in Section~\ref{sec:evaluation:querynoise}.
\rev{R2D7}{
    Another notable observation is that \name's advantage over competitors remains substantial also on extremely long time series, with a 22x speedup over MASS and a $\sim$100x speedup over other methods on the Wind dataset.
    This is relevant, as such datasets are becoming increasingly common in domains such as IoT and sensor data~\cite{wind_dataset}.
}
All in all, we conclude that \name is significantly faster than all competitors across various datasets with different characteristics and different domains, confirming its robustness and general applicability.
    
\subsubsection{Effect of query difficulty}\label{sec:evaluation:querynoise}
Pruning-based algorithms rely on the "left-tail" assumption~\cite{aumuller21}: only a small fraction of the dataset is similar to the query.
Their performance depends on the \emph{relative contrast} between the query and indexed data -- the ratio between distances to the closest and farthest indexed entities~\cite{aggarwal01,ceccarello2025hardnesss}.
Following~\cite{azizi23_elpis,zoumpatianos18}, we test algorithm sensitivity by varying noise levels and using out-of-distribution (OOD) queries from held-out data.
Figure~\ref{fig:row2}a shows results for Stocks and Weather, revealing that query times are inversely proportional to relative contrast (grey bars on secondary y-axis), consistent with previous work~\cite{azizi23_elpis,zoumpatianos18,ceccarello2025hardnesss}.
Normalized subsequence queries have lower contrast than raw ones, making them more challenging.
For high-noise or OOD queries, \name's pruning power weakens to match MASS (plus tree traversal overhead), suggesting sequential scans are more appropriate as they don't rely on pruning.
This motivates future work on a hybrid approach that selects between \name and MASS based on estimated query difficulty.

\begin{figure*}[tb]
    \begin{subfigure}[t]{0.49\linewidth}
        \centering
        \includegraphics[width=\textwidth]{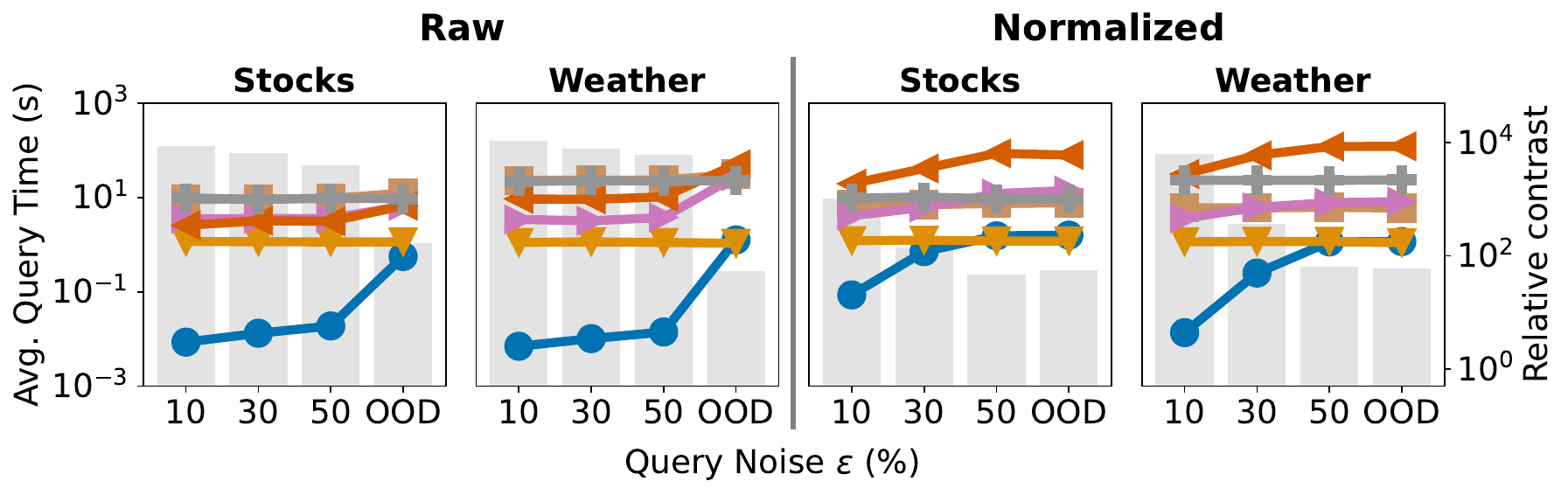}
    \end{subfigure}
    \begin{subfigure}[t]{0.249\linewidth}
        \centering
        \includegraphics[width=\textwidth]{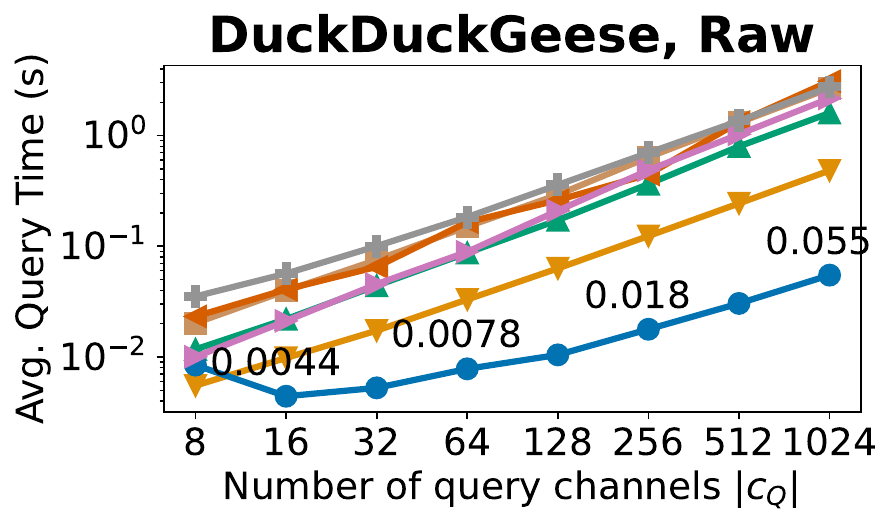}
    \end{subfigure}
    \begin{subfigure}[t]{0.249\linewidth}
        \centering
        \includegraphics[width=\textwidth]{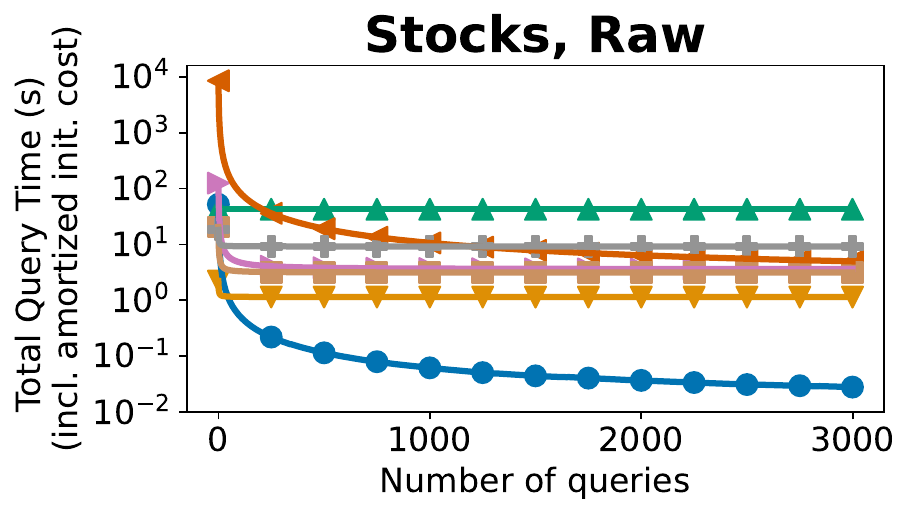}
    \end{subfigure}
    \begin{minipage}[b]{\linewidth}
        \centering
        \includegraphics[width=.7\textwidth]{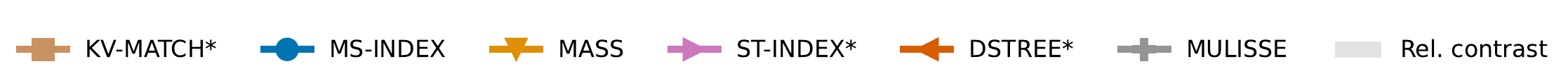}
    \end{minipage}
    \caption{Impact of query workload on algorithms. (a) Query time of algorithms with varying levels of query noise; (b) Query time over queried channels; (c) Amortized initialization time over queries. All y axes are presented in log scale.}
    \label{fig:row2}
        \centering
        \begin{subfigure}[t]{0.49\textwidth}
            \centering
            \includegraphics[width=\textwidth]{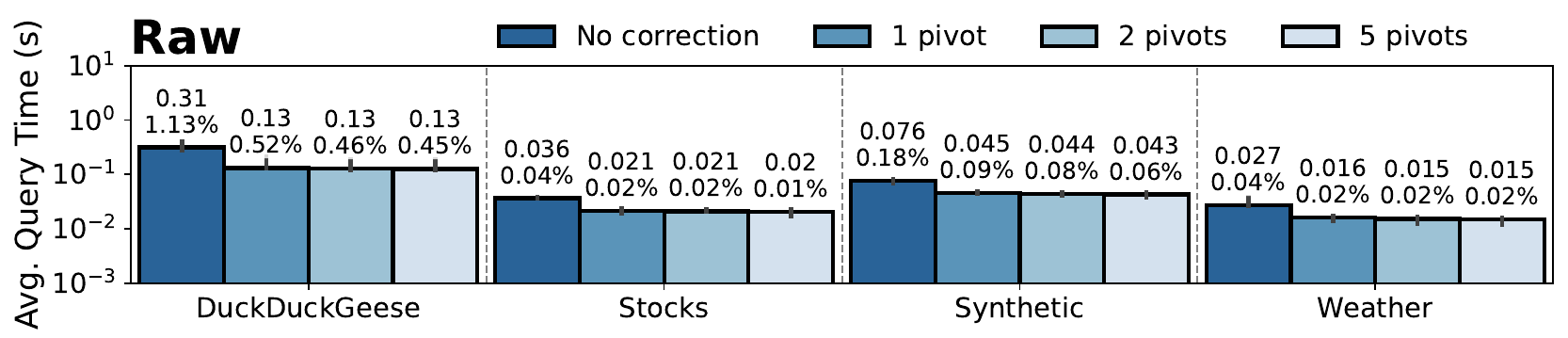}
        \end{subfigure}
        \begin{subfigure}[t]{0.49\textwidth}
            \centering
            \includegraphics[width=\textwidth]{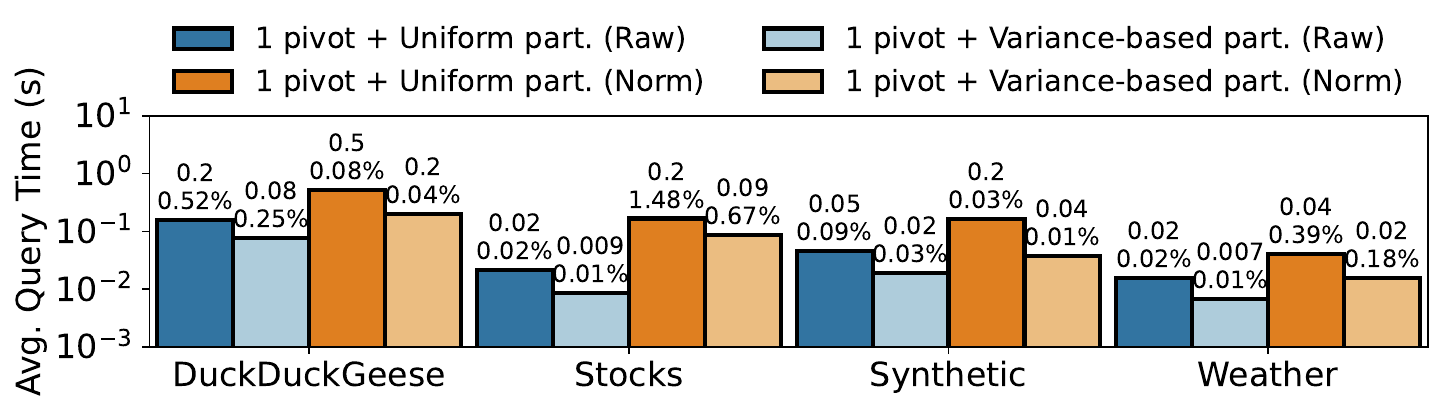}
        \end{subfigure}
        \caption{Impact analysis of the optimizations on \name: (a) Query time over the number of pivots in DFT-bound correction; (b) Query time of weighted partitioning vs. uniform partitioning.
        \rev{R1D8,R1D11}{Annotations indicate the average query time and the percentage of subsequences exhaustively considered (the rest is pruned).}}
        \label{fig:row3}
\end{figure*}

\begin{table}[tb]
    \centering
\caption{\rev{R1D10}{Relative contrast of queries and percentage of nodes pruned by \name, over different numbers of query channels on the DuckDuckGeese dataset (UEA).}}
\label{tab:pruning_channels}
\small
\begin{tabular}{ll|rrrr}
    \hline
    \multicolumn{2}{l|}{\textbf{Query channels} $|c_Q|$} & \textbf{16} & \textbf{64} & \textbf{256} & \textbf{1024} \\ 
    \hline
    \multirow{2}{*}{Rel. contrast} & Raw & 62 & 107 & 110 & 120 \\
    & Normalized & 14.5 & 14.4 & 13.9 & 13.9 \\
    \hline
    \multirow{2}{*}{Perc. pruned} & Raw & 70.54\% & 79.36\% & 92.20\% & 97.94\% \\
    & Normalized & 89.58\% & 84.66\% & 83.68\% & 83.61\% \\
    \hline
\end{tabular}
\end{table}

\subsubsection{Number of query channels}\label{sec:evaluation:query_channels}
We investigate algorithm sensitivity to the number of query channels using the \emph{DuckDuckGeese} dataset from UEA, containing 1024 channels. 
Figure~\ref{fig:row2}b shows \name's query time scales sublinearly with the number of channels, while other methods scale linearly.
This improvement comes from increased pruning power outweighing the cost of additional distance computations.
For instance, query time decreases from 8.4 to 4.4 ms between 8-16 channels, then grows slowly as marginal selectivity diminishes.
Table~\ref{tab:pruning_channels} confirms that \name prunes 70.5-97.9\% of R-tree nodes as channels increase from 16 to 1024 when querying for raw subsequences, demonstrating effective use of multivariate information.\footnote{The pruning percentage of nodes is not the same as the pruning power (the ratio between number of considered subsequences over all subsequences, around 99.9\% for \name), since many nodes in the tree contain groups of subsequences.}
In contrast, \dstree, \stindex, and \kvmatch suffer with more query channels as their union of per-channel results grows almost linearly.
\rev{R1D10}{
    For normalized subsequences, \name maintains superior performance but experiences a decreasing pruning power due to the curse of dimensionality -- as the number of channels increases, distances between points become more uniform, reducing discriminative power~\cite{rstar,marimont79}.
    This effect is less pronounced for raw subsequences, where the scale of values in a channel can still provide additional discriminative information that can help in pruning.
    The relative contrast values in Table~\ref{tab:pruning_channels} confirm this: increasing with $|c_Q|$ for raw subsequences while slightly decreasing for normalized ones.
}

\subsubsection{Amortizing the initialization cost}\label{sec:evaluation:amortized}
Since algorithms have different initialization costs, their relative efficiency depends on the number of queries executed.
Figure~\ref{fig:row2}c shows the amortized total cost (initialization + query time) up to 3000 queries.
While MASS is fastest for few queries due to its low initialization cost, \name becomes more efficient after just 45 queries, with its advantage growing with more queries due to its superior query performance.

\subsubsection{Effect of DFT bound correction optimization}\label{sec:evaluation:dft_bound}
Recall from Section~\ref{sec:optimizations} that we tighten the distance bound using a triangle-inequality-based correction term, enabling more pruning at the cost of added initialization overhead.
We evaluate this correction's impact on both initialization and query times with varying number of pivots.
Figure~\ref{fig:row3}a shows that adding just a single pivot improves query performance by a factor of 2, while additional pivots yield only marginal improvements.
\rev{R1D9}{
    This is known as the \emph{coverage saturation effect}, where the effectiveness of pivot-based bounding stabilizes as the space becomes increasingly well-covered with more pivots~\cite{chavez01,baeza1994proximity}.
    In our case, diminishing returns after the first pivot indicate that the space of remainders (i.e., the high-frequency DFT coefficients) has low complexity: the data is concentrated around a single point.
    This observation is not surprising, as -- by definition -- these remainders comprise lower-energy frequencies that mostly represent noise following a normal distribution~\cite{kay93}.
}
Nevertheless, the improvement over not correcting shows that the remainders still contain valuable information.
Our analysis shows pivot comparisons increase initialization time by 15\%, 30\%, and 75\% for 1, 2, and 5 pivots respectively.
Given the query time improvements, a single pivot provides the best cost-benefit tradeoff.

\subsubsection{Effect of the optimization for tightening the MBRs}\label{sec:evaluation:partitioning}
We evaluate the optimization for tightening MBRs via weighted R-tree partitioning (Section~\ref{sec:optimizations}) against uniform partitioning.
The results in Figure~\ref{fig:row3}b show that weighted partitioning improves query time by 1.5-3x by allowing the R-tree to better adapt to dataset characteristics, with larger gains on high-dimensional datasets like DuckDuckGeese and Synthetic.
From these results we can conclude that the weighted partitioning strategy is beneficial.

%% file: sections/6.conclusions.tex
\section{CONCLUSIONS}\label{sec:conclusions}
We considered the problem of \rev{R1D5}{fixed-length subsequence search on MTS}, and proposed \name, an \rev{R1D1}{exact} algorithm that adaptively partitions and indexes a DFT-based feature space, allowing for a cheap pruning of over 99\% of the subsequences.
Our evaluation with 34 datasets demonstrated that \name outperforms the state-of-the-art by two orders of magnitude for both raw and normalized subsequences, and that it scales sublinearly with the amount of query channels.

%% file: sections/n.acknowledgements.tex
\begin{acks}
    This work has received funding from the Horizon Europe research and innovation programmes STELAR~(101070122), AI4Europe (10107\-0000), TwinODIS (101160009), ARMADA (101168951), DataGEMS (101188416) RECITALS (101168490), and by $Y \Pi AI \Theta A$ \& NextGenerationEU project HARSH ($Y\Pi 3TA-0560901$).
\end{acks}